\documentclass[letterpaper, 11pt]{article}

\usepackage[margin=1in]{geometry}
\usepackage[dvipsnames]{xcolor}
\usepackage[leqno]{amsmath}
\numberwithin{equation}{section}
\usepackage{amssymb}
\usepackage{amsthm}
\usepackage{mathtools}
\usepackage[pagebackref=true]{hyperref}
\hypersetup{
    colorlinks=true,
    linkcolor=red,
    citecolor=blue,
    urlcolor=teal
}
\usepackage{cleveref}
\usepackage{bm}
\usepackage{enumerate}
\usepackage{caption}
\usepackage{subcaption}
\usepackage{microtype}
\usepackage[authoryear]{natbib}
\usepackage{xurl}
\usepackage{authblk}
\usepackage{nameref}
\usepackage{float}

\makeatletter
\let\orgdescriptionlabel\descriptionlabel
\renewcommand*{\descriptionlabel}[1]{%
    \let\orglabel\label
    \let\label\@gobble
    \phantomsection
    \edef\@currentlabel{#1}%
    \let\label\orglabel
    \orgdescriptionlabel{#1}%
}
\makeatother

\theoremstyle{plain}
\newtheorem{theorem}{Theorem}[section]
\newtheorem{lemma}[theorem]{Lemma}
\newtheorem{proposition}[theorem]{Proposition}
\newtheorem{corollary}[theorem]{Corollary}
\newtheorem{claim}[theorem]{Claim}

\theoremstyle{definition}
\newtheorem{definition}[theorem]{Definition}
\newtheorem{example}[theorem]{Example}

\theoremstyle{remark}
\newtheorem{remark}[theorem]{Remark}

\usepackage[ruled, linesnumbered]{algorithm2e}
\crefname{algocf}{alg.}{algs.}
\Crefname{algocf}{Algorithm}{Algorithms}

\newcommand{\eps}{\varepsilon}

\allowdisplaybreaks

\title{Regulation of Algorithmic Collusion, Refined:\\ Testing Pessimistic Calibrated Regret}
\author{Jason D. Hartline\thanks{hartline@northwestern.edu} \quad Chang Wang\thanks{wc@u.northwestern.edu} \quad Chenhao Zhang\thanks{chenhao.zhang.rea@u.northwestern.edu}}
\affil{Northwestern University}
\date{}

\newcommand{\priceset}{\mathcal{P}}

\newcommand{\Transc}{\mathcal{T}}
\newcommand{\Hist}{\mathcal{H}}
\newcommand{\Algo}{\mathcal{A}}
\newcommand{\Mech}{\mathcal{M}}

\newcommand{\TrueEr}{R^T}
\newcommand{\EstEr}{\widetilde{R}^T}
\newcommand{\WstEr}{\overline{R}^T}
\newcommand{\AggEr}{\breve{R}^T}
\newcommand{\EExp}[2]{
\operatorname*{\mathbb{E}}_{#1}\left[#2 \right]
}
\newcommand{\PPr}[2]{\Pr_{#1}\left[#2\right]}

\newcommand{\pmax}{\overline{p}}
\newcommand{\pimin}{\underline{\pi}}

\newcommand{\ASGN}{\leftarrow}

\newcommand{\estplcost}{\tilde{c}}
\newcommand{\plcost}{c_{\ast}}

\usepackage{accents}
\newcommand{\ubar}[1]{\underaccent{\bar}{#1}}
\newcommand{\costrange}{[\ubar{c},\bar{c}]}

\newcommand{\distseq}{\bm{\pi}^T}
\newcommand{\allocseq}[1]{\bm{#1}^T}
\newcommand{\priceseq}{\bm{p}^T}

\newcommand{\argmin}{\operatorname*{argmin}}

\begin{document}

\maketitle
\begin{abstract}
We study the regulation of algorithmic (non-)collusion amongst sellers in dynamic imperfect price competition by auditing their data as introduced by \citet{hartline2024regulation}.

We develop an auditing method that tests whether a seller's pessimistic calibrated regret is low. 
The pessimistic calibrated regret is the highest calibrated regret of outcomes compatible with the observed data. This method relaxes the previous requirement that a pricing algorithm must use fully-supported price distributions to be auditable.
This method is at least as permissive as any auditing method that has a high probability of failing algorithmic outcomes with non-vanishing calibrated regret. 
Additionally, we strengthen the justification for using vanishing calibrated regret, versus vanishing best-in-hindsight regret, as the non-collusion definition, by showing that even without any side information, the pricing algorithms that only satisfy weaker vanishing best-in-hindsight regret allow an opponent to manipulate them into posting supra-competitive prices. This manipulation cannot be excluded with a non-collusion definition of vanishing best-in-hindsight regret.

We motivate and interpret the approach of auditing algorithms from their data as suggesting a per se rule. However, we demonstrate that it is possible for algorithms to pass the audit 
by pretending to have higher costs than they actually do.
For such scenarios, the rule of reason can be applied to bound the range of costs to those that are reasonable for the domain. 
\end{abstract}

\section{Introduction}
The prevailing practice of making pricing decisions with algorithms by sellers in competitive markets has drawn scrutiny from lawmakers and regulators for concerns about price collusion. For example, the US Department of Justice recently filed a lawsuit against RealPage \citep{dojrealpage2024}, a company providing algorithmic pricing software for landlords in the apartment rental market, for allegedly facilitating price collusion. As Attorney General Merrick Garland stated, ``We allege that RealPage's pricing algorithm enables landlords to share confidential, competitively sensitive information and align their rents. Using software as the sharing mechanism does not immunize this scheme from Sherman Act liability...'' the legal ground of the regulatory action is based on the argument that the algorithm provided by RealPage enables a covert communication channel for market participants to coordinate and maintain higher than competitive prices (a.k.a. supra-competitive prices). There is one important reason behind the strong emphasis of communication in this argument: In many jurisdictions such as the US, overt communication is a prerequisite to establishing legal liability of collusion. Collusion with the absence of communication is actually not illegal according to the current legal theory \citep{harrington2018developing, harrington2022effect}.

On the other hand, researchers recently \citep{calvano2020artificial,asker2023impact,banchio2023adaptive} have also discovered a more disturbing fact that popular reinforcement learning algorithms can learn to collude without explicit communication just by engaging in repeated market interactions. These forms of implicit collusion pose new challenges for the regulation of algorithmic collusion.

There are two doctrines of antitrust analysis in the US legal framework: per se rule and rule of reason \citep{hovenkamp2018rule}. The per se rule deems certain business practices, like price-fixing or market division, illegal without requiring further investigation into their actual competitive effects. For example, 
in response to the algorithmic collusion concern around RealPage, San Francisco has recently enacted the per se regulation that precludes any use of algorithms in rental pricing \citep{sfban2024}.  The rule of reason, on the other hand, involves a more thorough evaluation of business practices, assessing whether they unreasonably restrain trade by examining their purpose, effects, and the overall context of the market. See \Cref{sec:regulation_classification} for a more comprehensive categorization of currently proposed regulations of algorithmic collusion.

This paper considers the regulation of algorithmic collusion by auditing pricing data as proposed by \citet{hartline2024regulation}.  
In their proposal, they argue that it is feasible to require all sellers to deploy pricing algorithms to pass the audit. This means that it would not harm the market to deem those who do not pass as illegal. The regulation to require algorithms to pass such an audit falls within the per se doctrine.

They then provide details of the per se rule (see the next paragraph) and propose a non-collusion audit implemented by a statistical test on the data collected during the deployment of a seller's algorithm.

The non-collusion audit they propose is based on ideas and observations from the theory of online learning. At its core, they propose using \emph{calibrated regret} as a quantitative measure for how far from non-collusion a seller running pricing algorithms is. Informally, given a sequence of market conditions and pricing decisions made by a seller, calibrated regret measures how much she can be better off by further utilizing the information revealed from her current pricing decisions. Low calibrated regret indicates that the seller is already close to best responding to the market environment she faces, which implies that she is not colluding since best responding corresponds to competitive behavior. To empirically audit non-collusion of a seller from the collected data, they develop a statistically efficient test of low calibrated regret of the seller. Being unable to pass the test is a violation of their suggested per se rule. 

\subsection{Our Contributions}

We study the framework by \citet{hartline2024regulation} in light of the standard guideline of binary classification: minimizing the number of false positives and false negatives. This classification is also suggested by \citet[VI. A]{harrington2018developing} when designing a per se rule specifying a prohibited set of algorithms.

\begin{quote}
\begin{description}
    \item[false positives\label{enum:more_good}] A competitive algorithm that is unable to pass the audit is a false positive. We want to allow a range of competitive algorithms as broad as possible to pass the audit.
    \item[false negatives\label{enum:less_bad}] An algorithm that enables supra-competitive prices and passes the audit is a false negative. We want as many algorithms that enable supra-competitive prices to fail the audit as possible.
\end{description}
\end{quote}
In the light of the above guideline, our main contribution is three-fold.
First, we reduce false positives by providing an improved auditing method that drops the prerequisite that an algorithm must use fully-supported distribution to pass the audit, thus allowing a broader range of competitive algorithms to pass (\Cref{subsubsec:lessfp}). Second, \citet{hartline2024regulation} show the need of vanishing calibrated regret as the (non-)collusion definition by providing an example under which testing best-in-hindsight regret passes supra-competitive algorithms when one seller can have side information related to the demand. The construction relies strongly on the usefulness of the side information. We show the stronger insufficiency of testing vanishing best-in-hindsight regret to exclude supra-competitive prices even when neither seller has side information.
Thus, using vanishing calibrated regret as the non-collusion notion is important to reduce false negatives even without side information (\Cref{subsubsec:lessfn}).

Third, we argue that rule of reason can also be useful by demonstrating that there exist supra-competitive algorithms that could pass the audit with a high inferred cost even if configured with the true cost (\Cref{subsubsec:unknown_cost}).

\subsubsection{Fewer False Positives}\label{subsubsec:lessfp}
First, we improve in the direction of reducing \ref{enum:more_good}. We note that there is a significant set of good algorithms that, without modification, is not able to pass the audit of \citet{hartline2024regulation}.\footnote{Note: \citet{hartline2024regulation} give a procedure for modifying these algorithms so that they can pass the audit.}

Specifically, \citet{hartline2024regulation} assume that the seller's algorithm outputs a distribution of prices in each round. The actual price posted in each round is sampled from the output price distribution. The auditing method computes an estimated regret from a transcript of the pricing algorithm consisting of, in each round: 1) the actual price posted, 2) the observed demand for the posted price, and 3) the distribution of prices from which the actual posted price is drawn. There is one key requirement: For any pricing algorithm to be auditable, the price distribution in each round must have full support. In other words, in each round, every price level must be posted with at least some probability.
We view this full-support requirement as restrictive because:
\begin{enumerate}
    \item The seller might not want to post some prices. First, the seller could possess some side information (that the regulator does not know) that makes her prefer to avoid certain prices. Second, the seller could deliberately avoid some prices for non-technical reasons (e.g., posting 2.99 instead of 3.00, or avoiding the number 13, etc.). 
    \item In practice, the exact price distributions of the seller are often unavailable. Asking the seller to submit full price distributions can also be problematic due to privacy issues. To apply the auditing method proposed by \citet{hartline2024regulation}, a plausible alternative is to aggregate the prices in a given time window, and use the empirical distribution as the price distribution in that window. If the size of the window is appropriately chosen such that the change in the price distributions is small in that window (for example, when a learning algorithm with a small learning rate is used for pricing), then the empirical distribution can be a good approximation to the true price distribution. Of course, this empirical distribution need not be fully-supported.
    
\end{enumerate}
\begin{remark}
    \Cref{lem:aggregate_possible} and \Cref{cor:aggregate_sample_complexity} present a formal statement of why auditing by aggregating prices to approximate distributions works with small learning rates. An interesting open question is to design a test for small learning rates from data so that the need of price distributions can be reliably removed for sellers with such learning rates.
\end{remark}

In this work, we propose a refined auditing method that enables the auditing of algorithms that do not use fully-supported price distributions. The refined auditing method continues to use an unbiased estimator for the counterfactual allocations, but it also maintains a \emph{pessimistic estimation} for the prices that are not in the support of the price distributions. The new method relaxes the previous requirement that a pricing algorithm must use fully-supported price distributions to be auditable and enables the seller to pass the test by demonstrating that her \emph{pessimistic calibrated regret} is low. The pessimistic regret is the highest regret of counterfactual outcomes compatible with the observed data.

In addition, we also clarify the interpretation of the non-collusion definition. Informally, the calibrated regret estimated both in \citet{hartline2024regulation} and our work is only the regret of a particular run of the seller's algorithm. The estimated regret is not of the algorithm itself. Therefore, the non-collusion definition is consequentialist and not deontological.\footnote{Consequentialism is the perspective of judging an act only by its outcome, while deontology is the perspective of judging an act by its inherent property.} We argue that this philosophy is appropriate not only from a legal perspective but also from a technical perspective. See \Cref{rmk:full_randomness}.

\subsubsection{Fewer False Negatives}\label{subsubsec:lessfn}
Second, we consider \ref{enum:less_bad}. The calibrated regret that \citet{hartline2024regulation} propose as the measure of distance to non-collusion is a strong notion of regret. A common weaker notion is best-in-hindsight (a.k.a.\,external) regret. Calibrated regret compares the performance of the chosen actions to a counterfactual scenario where the learner may switch among the actions using an arbitrary mapping. Best-in-hindsight regret, on the other hand, compares the performance of the chosen actions to the performance of the best \emph{fixed} action in hindsight. They argue that the stronger notion of vanishing calibrated regret is indispensable to exclude certain supra-competitive prices by giving a simple example where one seller has side information about buyers' valuations. The seller can utilize her side information to collude with the other seller and have non-positive best-in-hindsight regret. However, she would still have positive non-vanishing calibrated regret in this case.
The previous construction relies strongly on the usefulness of side information. Intuitively, regulation using calibrated regret is necessary in this case because the seller must best respond to such useful information. In this paper, we give a stronger argument for the need of calibrated regret for reducing false negatives. We show that even \emph{without any side information}, a large family of pricing algorithms that minimize the weaker notion of best-in-hindsight regret can be \emph{manipulated into playing supra-competitive prices}, and this is undetectable by a regulatory definition of non-collusion that uses best-in-hindsight regret. In other words, best-in-hindsight-regret algorithms can be susceptible to manipulation, thus, to prevent collusion and reduce false negatives, a regulator may want to discourage the use of algorithms that only satisfy this weaker property regardless of what side information can exist.

Before we describe the argument, we must be clear on what kind of manipulation we are concerned about. In the context of collusion auditing, we say a learning algorithm is \emph{manipulated into playing supra-competitive prices}, if:
a) the manipulator uses a strategy that leads the learning algorithm into an outcome where both of them have prices and payoffs exceeding correlated equilibria\footnote{Recall that a \emph{correlated equilibrium} is a joint distribution over players' actions, such that for every player $i$ and every action $a_i$, playing $a_i$ is a best response for $i$ conditioned on seeing $a_i$, given that everyone else will play according to the distribution.}, \emph{and} b) the manipulation cannot be excluded with the regulatory definition of non-collusion (e.g., no best-in-hindsight regret).

Note that we are only concerned about manipulative behavior that the regulatory non-collusion definition cannot exclude. The presence of such unexcluded manipulative behavior implies that certain algorithms can enable supra-competitive prices and evade the regulation, resulting in false negatives. On the other hand, manipulative behavior that have already been excluded by the definition do not result in false negatives.

Our technique is inspired by works such as \citet{braverman2018selling} and \citet{deng2019strategizing}
showing the susceptibility to manipulation of best-in-hindsight regret minimization algorithms. The manipulator in the example of \cite{deng2019strategizing} has non-vanishing best-in-highlight regret, but in our construction, the manipulator has vanishing best-in-highlight regret. Our construction is a strengthening with respect to their construction in this sense. 

More specifically, we construct an instance of the imperfect price competition with two sellers. In our example, the process that generates buyers' valuations is stationary and neither seller has any side information about the valuations of the buyer (or any other side information). One seller using a mean-based learning strategy for minimizing best-in-hindsight regret can still be manipulated into maintaining prices above equilibrium level for a significant number of rounds. Our construction has two important features. First, \emph{both} sellers play a constant fraction of rounds with supra-competitive prices in expectation, and \emph{both} extract revenue higher than correlated equilibrium, meaning the process is non-negligible and benefits both parties. Second and more importantly, \emph{both sellers have vanishing best-in-hindsight regret}, meaning such supra-competitive outcome would not be prevented by a per se regulation requiring vanishing best-in-hindsight regret.

\subsubsection{Unknown Costs: Per Se Rule v.s. Rule of Reason}\label{subsubsec:unknown_cost}
Third and finally, we consider the effects of not knowing the seller's cost in the auditing process. In the framework by \citet{hartline2024regulation}, the auditor knows the range of the seller's cost but not the exact cost. By their definition, as long as there exists some cost $\plcost$ within the range for which the seller's regret is low, the seller is deemed competitive. This leads to the following question: Are there natural algorithms that when configured with the \emph{true} cost $c$, find outcomes that are considered competitive by auditing methods for a higher cost $c'$, while actually being supra-competitive? We provide an affirmative answer to the question by examining sellers using the Q-learning algorithm \citep{banchio2022artificial} that converge to supra-competitive outcomes in simulation environments. The auditing on the seller correctly shows high estimated calibrated regret when a small and precise range of seller's cost is given. However, when the cost range goes far beyond the seller's actual cost, the estimated regret becomes low.

A legal implication of this result is that it may not be sufficient to solely apply the per se rule of auditing to regulate algorithmic collusion, particularly when the regulator only has limited information about the seller's cost upfront. In such settings, the regulator should further investigates the reasonable cost of the seller in the market context, which is an application of rule of reason.
\subsection{A Categorization of Proposed Regulations of Algorithmic Collusion}
\label{sec:regulation_classification}

As we have mentioned previously, we interpret proposed regulations of algorithmic collusion and motivate ours under the legal framework of antitrust analysis: per se rule v.s.\;rule of reason. In this section, we provide a more comprehensive analysis of those proposed methods and categorize them into the two doctrines. The conclusion of this subsection is summarized in \Cref{tab:comp_perse_reason}. More comparisons of technical aspects and other related work are included in \Cref{subsec:related_work}.

\begin{table}[ht]
    \small
    \centering
    \begin{tabular}{|c|c|}
    \hline
         \textbf{Examples of Regulation} & \textbf{Classification} \\
         \hline 
         outright ban on algorithms (e.g., \citet{sfban2024}) & per se \\
        \hline
         setting up formal non-collusive specifications and requiring proof of verification  & per se \\
         \hline
          inspecting source code to judge the degree of (non-)competitiveness within context &  rule of reason \\
        \hline
         requiring passing dynamic unit tests with specific protocols & per se \\
         \hline
         use context-specific input to dynamically test and analyze pricing behavior & rule of reason \\
         \hline
        requiring monitor wrapper \;\citep{chassangRegulatingCollusion2023} & per se \\
         \hline
        requiring passing data audit (\citep{hartline2024regulation} and this paper)
         & per se\\
         \hline
         analyzing cost under context and audit (\citep{hartline2024regulation} and this paper) & rule of reason \\
         \hline
    \end{tabular}
    \caption{Comparison of different proposed regulations of algorithmic collusion}
    \label{tab:comp_perse_reason}
\end{table}

The San Francisco ban on automated rent-setting \citep{sfban2024} is a per se rule that considers all algorithmic pricing on the rental market illegal. Clearly, such regulation stifles innovation by excluding rent-setting algorithms that promote competition and improve market efficiency.

\citet{harrington2018developing} discuss the approaches of \emph{static checking} an algorithm's source code\footnote{Here ``source code'' refers to any formal representation of an algorithm, which also includes, for instance, the architecture and weights of neural network models.} without running it and \emph{dynamic testing} the algorithm\footnote{Here ``algorithm'' refers to the executable form of the algorithm that a seller is going to deploy.} with synthetical inputs to learn its properties related to price collusion and conclude whether it should be prohibited. However, they suggest that to what extent the prohibition comes from a per se rule or rule of reason depends on details: Prohibitions from clear collusion-identifying properties checkable with these approaches can be classified as per se rules. Otherwise, classification into rule of reason is more appropriate.

We list a few proposals using static checking or dynamic testing and classify them. 

Static checking approaches range from having human experts manually review the source code to automatically generating machine-checkable proofs or counterexamples. 
One possible per se regulation of algorithmic collusion with static checking is to use the technique of static  verification with formal logic: 
The regulator sets up formal \emph{specifications} that exclude collusion-identifying properties and requires the sellers to submit the proofs that their algorithms are verified against the specifications. Algorithms without a proof of verification are prohibited. 
On the other hand, if the regulator inspects the source code of pricing algorithms to judge its (non)-competitiveness within a specific context, such an approach is in line with the rule of reason.

Dynamic testing an algorithm involves running the algorithms with certain inputs. One possible per se rule the regulator might use is to prepare a set of test cases and announce a uniform protocol for conducting tests of all algorithms of sellers. The protocol will specify how the test is conducted (e.g., by having the seller send a representative taking the algorithm to a brick-and-mortar site of the regulator agency or uploading the algorithm to a cloud service) and pass-fail criteria.  The regulator prohibits algorithms that fail the test following the protocol. Alternatively, the regulator can analyze the pricing behavior of each seller's algorithm by feeding it with context-specific inputs. Such regulation would be aligned with the rule of reason.

\citet{chassangRegulatingCollusion2023} propose another regulation based on the relation of regret and collusion observed in \citet{chassang2022robust}. The regulation approach they propose requires the seller to attach a monitor algorithm, known as a ``pricing wrapper,'' to the pricing algorithm to ensure that the composition satisfies certain no-regret properties. This approach can also be interpreted as a per se rule because it prohibits any algorithm that cannot be obtained by the composition of an existing pricing algorithm and the pricing wrapper.

Finally, our proposal of auditing from data provides a clear prohibition determination without requiring access to the algorithms themselves while making minimal assumptions on the input space of the algorithms. While it is a per se rule, it can be augmented with the rule of reason analysis to narrow the range of what costs are reasonable.

\subsection{Related Work}\label{subsec:related_work}
\paragraph{Algorithmic Collusion}
 Most papers studying algorithmic collusion from technical perspectives consider the Q-learning algorithm \citep{watkins1989learning}, a common reinforcement learning algorithm. \citet{calvano2020artificial, klein2021autonomous, asker2022artificial,asker2023impact, banchio2022artificial,banchio2023adaptive} study Q-learning under various settings with simulations and theoretical analysis. They consistently report that Q-learning can find and maintain, without explicit communication, supra-competitive prices (or infra-competitive bids) when in competition with each other. A few other papers have also explored algorithmic collusion beyond Q-learning, such as UCB \citep{hansen2021frontiers} and large language models \citep{fish2024algorithmic}. Our simulation in \Cref{sec:unknown_costs} follows the setup of \citet{banchio2022artificial}. These empirical and theoretical findings and concerns are among the main motivations of our work.

\paragraph{Legal Landscape of Anti-collusion Analysis}
US statutes regulating price collusion were enacted more than a hundred years ago, long before the era of digital markets and algorithmic pricing. They include the Sherman Act (1890), the Federal Trade Commission Act (1914), and the Clayton Act (1914). The recent court cases such as \cite{InreTextMessagingAntitrustLitigation2015, 1993brooke} interpreting these statutes for price collusion have affirmed the jurisprudence of requiring express agreement as the prerequisite of establishing liability. 

The Sherman Act literally prohibits acts that are ``in restraint of trade and commerce'' without clarifying how it should be applied \citep{sawyer2019us}. In early cases such as \citet{1899addyston} and  \citet{1918chicago}, the language of the statute is interpreted as applicable to any restraint of trade, which constitutes the per se mode of analysis.
The rule of reason doctrine in antitrust first appeared in the US Supreme Court ruling of \citet{1911standard}. Led by the then Chief Justice Edward White, the court decided that the Sherman Act should be ``construed in the light of reason,'' hence only applies to \emph{unreasonable} restraints of trade. Over the years, the court has narrowed the domain of per se rules in traditional antitrust cases while incorporating more analysis informed by economic principles to the application of rule of reason. \citet{sawyer2019us} and \citet{gavil2011moving} discuss the evolution of the two doctrines. \citet{hovenkamp2018rule} discuss the scope where rule of reason analysis should be applied in the non-algorithmic antitrust settings. In light of the new challenges posed by algorithmic collusion, \citet{harrington2018developing} propose adding per se prohibition for certain algorithms to competition laws. We motivate and interpret our work within the legal framework proposed by \citet{harrington2018developing}.

\paragraph{Regulation of Algorithmic Collusion}

In addition to the auditing approach proposed in \citet{hartline2024regulation}, \citet{harrington2018developing} and \citet{chassangRegulatingCollusion2023} discuss alternative proposals of regulating algorithmic collusion. See \Cref{sec:regulation_classification} for a legal perspective for on this works.

The regulation of requiring non-regret monitor wrapper proposed by \citet{chassangRegulatingCollusion2023} is based on the observations in \citet{chassang2022robust}, which consider the problem of screening non-algorithmic collusion in procurement auctions. Similar to the auditing approach proposed in \citet{hartline2024regulation}, \citet{chassang2022robust} estimate the demand functions from data and use the demand functions to compute regret-like quantities. However, there are several differences. The framework we consider makes minimal assumptions on the demand functions that a seller faces, namely, the demand is between $[0,1]$ and monotonically non-increasing in prices. We estimate the demand functions that a single seller faces utilizing the randomization of seller's algorithm without knowledge of other sellers' strategy. In comparison, \citet{chassang2022robust} consider the estimation problem when the form of the demand functions is known from the auction format and the bids of all bidders are available in the data. To deal with buyer distributions with imperfect competition, their approach would need assumptions on the demand while ours does not. Instead of assuming fixed production cost across rounds, \citet{chassang2022robust} also consider the case when the cost of a seller can be different for each round. Therefore, when computing the regret, the deviation of a seller's strategy has the form of changing the prices across each round proportionally. But in our framework, the deviation can be arbitrary.

\paragraph{Exploitation of a No-regret Learner}\citet{braverman2018selling}  consider the problem of repeated selling of an item to an agent using no-regret learning algorithms. They propose the notion of mean-based algorithms and show that mean-based algorithms guaranteeing no best-in-hindsight regret can be manipulated by a seller to extract full surplus. \citet{deng2019strategizing} consider the problem of manipulating no-best-in-hindsight-regret learner in general 2-player bimatrix games to get beyond the Stackelberg payoff\footnote{Stackelberg payoff is the payoff in a \emph{Stackelberg equilibrium}. Recall that in a Stackelberg equilibrium, a player (called leader) \emph{commits} to a strategy such that when the other player (called follower) best responds to the strategy, the payoff of the leader is maximized.}. Our example in \Cref{sec:external_manipulable} is inspired by their work but the construction is tailored for a dynamic price competition game.

Concurrently and independently, \citet{arunachaleswaran2024algorithmic} provide an example (referred to as ``their result'' thereafter) where two sellers could be supra-competitive compared to correlated equilibria when \emph{one} seller is using a vanishing calibrated regret algorithm. Their result is different from our result in \Cref{sec:external_manipulable} in the following two ways.

First, we take a different perspective on collusion and the results convey different messages. In \cite{arunachaleswaran2024algorithmic}, their concern is about the \emph{cause} of algorithmic collusion, and they claim that their result shows algorithmic collusion can arise even if there is no ``threat'' between the sellers. On the other hand, our paper's question is about the \emph{exclusion} of algorithmic collusion. We show that using no-best-in-hindsight regret as the non-collusion definition is not appropriate as it does not exclude some collusion even when side information is guaranteed not to exist. Although as is shown in \citet{arunachaleswaran2024algorithmic}, a seller running vanishing-calibrated-regret algorithm can still be manipulated to supra-competitive prices, this scenario will be prevented by our audit as the manipulating seller will not have non-vanishing calibrated regret and would be caught by the regulator.

Second, even in a pure technical sense, our results are different from theirs.
In their construction, only one seller is guaranteed to have no regret (calibrated regret or best-in-hindsight regret), while in our example, both sellers have vanishing best-in-hindsight regret. In fact, it is impossible to guarantee that both sellers have no calibrated regret while they post supra-competitive prices. Our results reaffirm the importance for \emph{both} sellers to have vanishing calibrated regret for competitive prices in a different way from them.

\subsection{Limitations}
Finally, we discuss a few caveats of our framework in this subsection.

First, if the seller uses pricing algorithms that do not optimize, she could also have high regret. Therefore, she would not pass our audit despite not being competitive. However, we argue that this is not a significant issue from the perspective of a per se rule. 
Sellers employ algorithmic pricing for the purpose of maximizing their profit and there are a variety of optimizing algorithms readily available. Therefore, using an algorithm that does not optimize is unlikely to serve any reasonable business purpose for a profit-seeking seller.

Second, as is discussed in \Cref{subsubsec:unknown_cost}, when the cost information of the seller possessed by the regulator is limited, the range of plausible costs can be large. In this situation, there exist algorithms with correctly configured true costs that exhibit supra-competitive behavior while still being able to pass the audit due to the existence of a plausible cost with low regret.  Therefore, the imprecision of cost information might significantly impact the auditing result and presents challenges to regulating collusion from pricing data. %
A mitigation to this problem is applying rule of reason: A further investigation of the sellers and market contexts is needed to narrow down the cost range.\footnote{Note: Third-party algorithms could be required to log the costs that they are configured with.}

Third, the current theoretical analysis gives a large sample complexity bound on our auditing method. Therefore, it is more applicable for regulating algorithmic collusion in online market environments with high-frequency transactions and a large amount of data. An open question is to improve the sample complexity of our auditing method.

\section{Preliminaries}\label{sec:prelim}
We consider a setting where $n$ sellers repeatedly compete for selling a good in $T$ rounds. Seller $i$ has cost $c_i \in \costrange$. In each round $t$, seller $i$ posts a price $p_i^t \in \priceset$ where $\priceset$ is a $k$-element set of possible price levels.

Let $\pmax = \max\{p:p \in \priceset\}$ be the maximum possible price level. Given all the sellers' prices, the demand (a.k.a.\;allocation) for seller $i$ is $x_i^t:\priceset^n \to [0, 1]$. We assume that fixing the prices $\bm p_{-i}$ posted by the sellers other than $i$, the allocation $x_i^t(p_i, \bm p_{-i})$ is non-increasing in $p_i$. Seller $i$'s utility at round $t$ posting $p$ is $u_i^t(p)=(p-c_i)x_i^t(p, \bm p_{-i}^t)$. At the end of each round, the seller gets her allocation as the feedback.\footnote{She might also get other information, but we as auditors cannot directly observe other information.} This is known as \emph{bandit feedback} in the literature of online learning. Moreover, the demand can be arbitrary and even adversarial under our framework.

The problem seller $i$ faces is an online-learning problem. Seller $i$'s strategy in round $t$ can be represented as a price distribution $\pi_i^t \in \Delta(\priceset)$, where $\Delta(\priceset)$ is the set of distributions over $\priceset$. She posts prices $p_i^t$ according to the distribution $\pi_i^t$ and obtains the utility resulting from posting $p_i^t$. 

The seller's behavior in a sequence of rounds of competitions can be summarized as a \emph{transcript}. As is the only feedback the regulator can assume the seller gets at the end of each round, the transcript contains the allocation $x_i^t(p_i^t)$ corresponding to the price the seller posted, but not the full demand function $x_i^t(\cdot)$.
\begin{definition}
    Call $\Transc_i^t = \{x_i^s(p^s), p_i^s, \pi_i^s\}_{s=1}^t$ where $p_i^s \sim \pi_i^s$ a \emph{transcript} of length $t$ for seller $i$. The set of all the length-$t$ transcripts for seller $i$ is denoted as $\Hist^t$.
\end{definition}
As an auditor, given the transcript of the seller, we want to test whether the seller is exhibiting (non-)collusive behavior. 

We will consider the non-collusion definition from \citet{hartline2024regulation}, which is a unilateral property. Therefore, from now on we focus on a single seller's behavior, and we will drop the subscript $i$ whenever possible. We denote the sequences $\allocseq{x}:=\{x^t\}_{t=1}^T$, $\priceseq := \{p^t\}_{t=1}^T$, and $\distseq:=\{\pi^t\}_{t=1}^T$.

\citet{hartline2024regulation} propose that the seller is non-collusive if the transcript satisfies the vanishing calibrated regret property. We define calibrated regret and vanishing calibrated regret as follows.
\begin{definition}\label{def:swap_regret}
    Given the ground-truth $\allocseq{x}$ and seller's cost $c$, the \emph{calibrated regret} of the \emph{transcript} for a seller with cost $c$ is
    \[ 
        \TrueEr(\allocseq{x}, c) = \max_{\sigma: \priceset \to \priceset}\frac{1}{T}\sum_{t=1}^T\EExp{p \sim \pi^t}{u(\sigma(p^t),x^t)-u(p,x^t)}.
    \]
    The seller's calibrated regret is called \emph{vanishing} if $\lim_{T \to \infty} \TrueEr(\allocseq{x}, c)=0$.
\end{definition}
Unless noted otherwise, we use ``regret'' to refer to ``calibrated regret.''
Since the auditor does not know the true cost $c$, as long as there exists a \emph{plausible cost} $\plcost \in \costrange$ such that $\lim_{T \to \infty} \TrueEr(\allocseq{x}, \plcost) = 0$, the seller is considered \emph{plausibly non-collusive}. The auditing method that \citet{hartline2024regulation} provides is based on estimating the calibrated regret. However, for the auditing method to work properly and provide a meaningful guarantee, it imposes an \emph{auditability requirement} that for all $1 \leq t \leq T$, the price distribution $\pi^t$ must be fully-supported. Sellers using algorithms that have vanishing calibrated regret but do not satisfy this requirement are unable to pass the audit without modification.

We conclude this section with two remarks on the settings.

The first remark says that the $k$-element price set can be chosen by the seller without affecting our results. And the second remark clarifies the randomness under the definition of calibrated regret in our framework.
\begin{remark}\label{rmk:endo}
    The $k$-element price set can be \emph{endogenous} to our model, meaning it can be chosen by the seller instead of being prescribed in advance. For example, the regulator first announces a continuum of possible prices $[0, h]$ and parameters of the audit, and the seller decides to discretize the set into $k$ price levels according to her (learning) algorithm. Then the regulator can conduct an audit assuming all the prices in $[0, h]$ can be posted, and there will be an additional ``discretization loss'' in the regret estimation. See \Cref{rmk:step4prime} for details of handling this in our auditing method.
\end{remark}

\begin{remark}\label{rmk:full_randomness}
    We note a technical subtlety in \Cref{def:swap_regret}. \Cref{def:swap_regret} is \emph{not} the calibrated regret of the \emph{algorithm} that generates the transcript. Consider the following example: A seller flips a coin to decide whether to run a collusive algorithm or a non-collusive algorithm. Then the quantity in \Cref{def:swap_regret} is either the calibrated regret of the collusive algorithm, or that of the non-collusive algorithm, depending on the outcome of the coin flip. The practical implication of this subtlety is contained in the following question: Suppose the auditor has evidence that the seller's algorithm is doing such a coin flip, but the outcome is that she happens to run the non-collusive branch, does he flag the seller? \Cref{def:swap_regret} implies that our answer is no, hence our definition of non-collusion is \emph{consequentialist} and not \emph{deontological}. In fact, we can show that in the realm of algorithmic collusion, we are only \emph{able} to use a consequentialist definition, because a deontological collusion detection using only observable data is technically impossible. See \Cref{appendix:full_randomness} for the technical details.

\end{remark}

\section{A Framework of Auditing Methods}\label{sec:regret_estimation}

In this section, we present a framework that defines a property called \emph{consistency} which describes that an auditing method \emph{correctly} audits an algorithm (\Cref{def:relaxed_consistency}). Although the auditing method that \citet{hartline2024regulation} propose satisfies a more restrictive consistency property (\Cref{def:consistency}), it relies on the \emph{full support requirement}, which means that the pricing algorithm must use every price with non-zero probability. When auditing algorithms that may not randomize over prices with full support, there is missing information because outcomes for prices that are posted with zero probability cannot be estimated. We show that the more restrictive consistency property cannot be satisfied with missing information and this is why we relax it to \Cref{def:relaxed_consistency}. We refer to \Cref{def:relaxed_consistency} as the \emph{one-sided} consistency requirement and \Cref{def:consistency} as the \emph{two-sided} consistency requirement. Then, under the one-sided consistency requirement, we define the notion of \emph{pessimistic (counterfactual) regret} that uses conservative upper bounds on the allocation when there is missing information (\Cref{def:worst_possible_regret}). Finally, we show that a correct auditing method under one-sided consistency must make decisions by considering that the regret of the seller is at least the pessimistic regret. This motivates the design of the improved auditing method in \Cref{sec:general_algo}.

\paragraph{From Auditing Methods to Regret Estimators}
First, we claim that it is without loss of generality to focus on regret estimators when studying auditing methods. This is done via a reduction argument. That is, if we have a probably approximately correct auditing method, then we also have a probably approximately correct regret estimator, and vice versa.

To begin we define auditing methods and regret estimators as follows:
\begin{definition}
    For a fixed cost $c$ and regret threshold $r>0$, an \emph{auditing method} is a mapping $\mathcal A:\bigcup_{t \geq 0} \Hist^t \to \{\mathrm G, \mathrm S\}$, and the output indicates the regret is greater (G) or smaller (S) than $r$ \emph{assuming} the seller's cost $c$. A \emph{regret estimator} is a mapping $\mathcal A:\bigcup_{t \geq 0} \Hist^t \to \mathbb R$, and the output indicates the estimated regret \emph{assuming} the seller's cost $c$.
\end{definition}

Note that the auditing method proposed by \citet{hartline2024regulation} is based on a regret estimator, so it suffices to reduce regret estimation to auditing.
\begin{lemma}\label{prop:reduction}
    Suppose we have an auditing method $\mathcal A$ that satisfies the following property:
    \begin{itemize}
        \item If the regret of a transcript $\Transc^T$ of length $T$ is smaller than or equal to $r$, then it outputs S (smaller), and
        \item if the regret of the transcript $\Transc^T$ is larger than or equal to $r+\eps$, then it outputs G (greater),
    \end{itemize}
    with probability at least $1-f(\eps, T)$. Then there exists an estimator of the regret of $\Transc^T$ up to accuracy $\eps$ and error probability at most $\frac{\pmax f(\eps, T)}{\eps}$.
\end{lemma}
All proofs, unless
otherwise specified, are deferred to \Cref{sec:oproofs}.

\paragraph{The Consistency Requirement}
From now on we focus on regret estimators. To do correct hypothesis testing with the regret estimator, we want the regret estimator to be approximately \emph{consistent}, defined below:
\begin{definition}[Consistency, one-sided]\label{def:relaxed_consistency}
    A regret estimator $\Algo$ is \emph{consistent} if
    for any sequence $\{\{x^t(p^t)\}_{t=1}^T,\priceseq,\distseq\}_{T\geq 1}$ of transcripts, $\eps>0$, cost $c$, and sequence of ground-truth sequence of allocations $\{\allocseq{x}\}_{T \geq 1}$ agreeing with the transcripts,
    \[ \lim_{T \to \infty} \PPr{\priceseq \sim \distseq}{\Algo(\Transc^T) < \TrueEr(\allocseq{x}, c) - \eps} = 0. \]
\end{definition}
The above definition says that a regret estimator must approximately output at least the true regret of the transcript in the limit.

Before we study the implication of \Cref{def:relaxed_consistency}, we explain why we only require a regret estimator to approximate an upper bound of the true regret instead of approximating the true regret itself---The following definition is tempting: 
\begin{definition}[Consistency, two-sided]\label{def:consistency}
    A regret estimator $\Algo$ is \emph{consistent} for a set of transcripts $S$ if, 
    for any sequence $\{\{x^t(p^t)\}_{t=1}^T,\priceseq,\distseq\}_{T\geq 1}$ of transcripts in $S$, $\eps>0$, cost $c$, and sequence of ground-truth sequence of allocations $\{\allocseq{x}\}_{T \geq 1}$ agreeing with the transcripts,
     \[ \lim_{T \to \infty} \PPr{\priceseq \sim \distseq}{|\Algo(\Transc^T) - \TrueEr(\allocseq{x}, c)| \geq \eps} = 0. \]
\end{definition}
Next, we explain why \Cref{def:consistency} is not appropriate in the general case (i.e. for sellers with not fully-supported price distributions). Although the regret estimator in \citet{hartline2024regulation} indeed satisfies the two-sided consistency property (\Cref{prop:previous_consistency}), it only works for transcripts with fully-supported price distributions. Unfortunately, the two-sided consistency requirement is too strong for algorithms that accept all the possible transcripts (\Cref{prop:consistency_negative}). In other words, there are pricing algorithms that produce transcripts for which the regret cannot be consistently (according to \Cref{def:consistency}) estimated.
\begin{proposition}\label{prop:previous_consistency}
   The algorithm in \citet{hartline2024regulation} is consistent (two-sided), for the set of transcripts satisfying $\underline{\pi}^T = \min_{t\leq T,p} \pi^t(p) = \omega(T^{(-1/4)})$, 
\end{proposition}
\begin{proposition}\label{prop:consistency_negative}
    No deterministic regret estimator is consistent (two-sided) for the set of all transcripts. In particular, there exists a seller who has vanishing true regret, but her regret cannot be consistently (two-sided) estimated.
\end{proposition}

The above proposition shows that it is not possible to get an estimator of the regret satisfying \Cref{def:consistency} in the general case. The auditor could be unable to certify a truly competitive seller. This is why we ask for a relaxed property in \Cref{def:relaxed_consistency} that the regret estimator must output an upper bound of the regret with high probability.

Returning to the discussion of the one-sided consistency requirement in the general case. Recall that our philosophy is that it is the seller's responsibility to demonstrate enough information that she is competitive. The one-sided consistency property ensures that missing information is properly accounted for so that a supra-competitive seller is never deemed as competitive because of the regret estimation.

Of course, the one-sided consistency requirement does not rule out regret estimators that always output trivial upper bounds of the regret. In the next section, we will provide an algorithm that outputs the least possible upper bound.

Finally, \Cref{def:relaxed_consistency} has an important implication: Whenever there is some missing information, any regret estimator satisfying \Cref{def:relaxed_consistency} must output at least the \emph{pessimistic (counterfactual) regret} of the transcript. We first define pessimistic (counterfactual) regret (\Cref{def:worst_possible_regret}) and then state the implication (\Cref{lem:indistinguishable}).

Intuitively, the pessimistic regret is the highest regret that outcomes compatible with the observed data can generate. We define such compatibility between outcomes and the observed data as follows.
\begin{definition}[Indistinguishable allocations]
    Fix the sequence of price distributions $\distseq$, and let $C^t = \{p \in \priceset : \pi^t(p) > 0 \}$ be the set of price levels that have non-zero probability being posted in round $t$. Two sequences of allocations $\allocseq{x}, \allocseq{z}$ are called \emph{indistinguishable} if $x^t(\cdot)$ and $z^t(\cdot)$ have the same support $C^t$ for every $1 \leq t \leq T$ and
    \[ x^t(p)=z^t(p) \quad \text{for all} \quad p \in C^t \quad \text{and} \quad 1 \leq t \leq T. \]
\end{definition}
The indistinguishability relation is an equivalence relation. If $\allocseq{x}, \allocseq{z}$ are indistinguishable, then there is no way to separate them from data.

With the definition of indistinguishable allocations, we define the pessimistic regret.
\begin{definition}[Pessimistic counterfactual regret]\label{def:worst_possible_regret}
    Fix the sequence of price distributions $\distseq$, and with a given sequence of allocations $\allocseq{x} = \{x^t\}_{t=1}^T$, the \emph{pessimistic (counterfactual) regret} is defined  as
    \[
    \WstEr(c, \allocseq{x}) = \sup\{ \TrueEr(c, \allocseq{z}): \allocseq{z} \text{ indistinguishable with } \allocseq{x} \}.
    \]
\end{definition}

The following proposition implies that, by only looking at the transcript, the auditor cannot rule out the possibility that the true regret is as high as the pessimistic regret.

\begin{proposition}\label{lem:indistinguishable}
    Any one-sided consistent regret estimator $\mathcal A$ must satisfy
    \[ \lim_{T \to \infty} \PPr{\priceseq \sim \distseq}{\mathcal A(\Transc^T) < \WstEr(\allocseq{x}, c) - \eps} = 0, \]
    for any $\eps>0$, cost $c$, transcript $\Transc^T$, and sequence of ground-truth sequence of allocations $\{\allocseq{x}\}_{T \geq 1}$.
\end{proposition}

Inspired by \Cref{lem:indistinguishable}, in the next section we present an auditing method that enables the seller to pass the test by demonstrating her pessimistic regret is low.

\section{Testing Pessimistic Regret}\label{sec:general_algo}
In this section, we present the refined auditing method that estimates the pessimistic regret (as defined in the previous section). We then show the following guarantee: With a sufficient amount of data, if the seller's pessimistic regret is low, then she passes the audit with high probability, and if the seller's pessimistic regret is high at every cost in $\costrange$, then she fails the audit with high probability.

First, we do a decomposition of calibrated regret so that we can compute it efficiently. Recall that the calibrated regret is defined to be the maximum benefit of deviation by doing a price swap $\sigma:\priceset \to \priceset$. A useful decomposition of calibrated regret is to first compute the benefit of changing price $p$ to $q$, then take the maximum over all possible $q \in \priceset$, and finally sum the result over all $p \in \priceset$. Formally, let
\[ \TrueEr_{p, q}(c, \allocseq{x}) = \frac 1T \sum_{t=1}^T \pi^t(p)\left[(q-c)x^t(q)-(p-c)x^t(p)\right], \]
then we have

\[ \TrueEr(c, \allocseq{x}) = \sum_p \max_q \TrueEr_{p, q}(c,\allocseq{x}). \]

The auditing method estimates the pessimistic regret based on the above decomposition and the pessimistic estimation of allocations from data. The steps are described at a high level as follows.\footnote{A formal pseudocode description can be found at the end of \Cref{sec:oproofs}.}

The input to the general auditing method contains the prices the seller posts in each round $\{p^t\}_{t=1}^T$, the allocations (demands) of the posted prices $\{x^t(p^t)\}_{t=1}^T$, seller's price distributions $\{(\pi^t)\}_{t=1}^T$, and the threshold $r$. The price distributions need not be fully-supported. The method proceeds in the following steps:
\begin{description}
    \item[Step 1] We estimate the allocation for every round using the transcript. For each round $t$, let $C^t:=\{p \in \priceset:\pi^t(p_j)>0\}$ be the support of the price distribution $\pi^t$. For every price $p \in C^t$, the propensity score estimator is used to estimate the allocation
    \[ \hat x^t(p) = \begin{cases}
                x^t(p)/\pi^t(p) & p=p^t, \\
                0 & \text{otherwise}.
            \end{cases} \]
    For the prices that are not in the support, we use the estimator of the allocation at the largest price $p'$ that is smaller than $p$ while being in the support.
    \[ \hat{x}^t(p):=\hat{x}^t(p') \quad \text{where} \quad p' = \max\{r \leq p:r \in C^t\}. \]
    If no such price exists, then the estimation is capped with 1.
    \item[Step 2] We estimate the true regret of the pessimistic allocation $\WstEr(c, \allocseq{x})$ with the estimator $\EstEr(c, \allocseq{x})$, built up from the estimator $\EstEr_{p, q}(c, \allocseq{x})$ for $\WstEr_{p,q}(c, \allocseq{x})$. Specifically, we first compute the benefit of substituting price $p$ with $q$
    \[ \EstEr_{p, q}(c) = \frac 1T \sum_{t=1}^T \pi^t(p)\left[(q-c)\hat{x}^t(q)-(p-c)\hat{x}^t(p)\right]. \]
    Then the pessimistic regret can be estimated by summing the highest benefit of changing each price
    \[ \EstEr(c) = \sum_p \max_q \EstEr_{p, q}(c). \]
    \item[Step 3] We minimize the pessimistic regret over all the possible costs to compute the pessimistic plausible regret:
    \begin{equation*}
        \min_{c \in \costrange} \quad \EstEr(c).
    \end{equation*}
    Note that this can be done in polynomial time even with a continuum of costs. In fact, following the observations of \citet{Nekipelov2015}, for each $p,q$, $\EstEr_{p, q}(c)$ is linear in $c$. Therefore, $\EstEr(c) = \sum_p \max_q \EstEr_{p, q}(c)$ is a convex (piece-wise linear) function of $c$, which can be efficiently minimized over the closed set $\costrange$.
    \item[Step 4] Finally, we compare the estimated plausible regret plus an additional error margin $\delta^T$ with the required threshold. The error margin ensures that the seller cannot pass the audit when the information revealed from the transcript is insufficient to guarantee the reliability of the regret estimator. Specifically, if $\EstEr(\estplcost) + \delta^T \leq 2r$,\footnote{If the choice of the price set $\priceset$ is endogenous (see \Cref{rmk:endo}), then we add a discretization loss to the estimated regret. Specifically, let the price levels be $q_1< \dotsb < q_k$ and $d=\max_{0 \leq i \leq k+1}(q_i-q_{i-1})$, where $q_0=0, q_{k+1}=h$. Output result by comparing $\EstEr(\estplcost) + \delta^T+d$ with $2r$ similarly as Step 4.\label{rmk:step4prime}} then we output PASS, and FAIL otherwise, where
    \[ \delta^T = \frac{k\pmax}{T}\sqrt{2\log\left(\frac{2k^2}{\alpha}\right)\cdot \sum_{s=1}^T\left(\frac{1}{\min_{p \in C^s}\pi^s(p)}+1\right)^2}. \]
\end{description}

The following theorem identifies the sample complexity of testing pessimistic calibrated regret.
\begin{theorem}[Sample complexity of testing pessimistic regret]\label{thm:sample_complexity} Consider the environment with maximum possible price level $\pmax$ and $k=|\priceset|$ price levels. Let $c_0$ be the seller's true cost and $
    \plcost = \arg\min_{c\in\costrange}\WstEr(c,\allocseq{x})$ be the plausible cost of the seller. Fix confidence level $1-\alpha$, threshold $r$ and let $\pimin=\min_{p\in C^t,1\leq t \leq T}\pi^t(p) $ be the minimum probability in the support of price distributions. With our refined auditing method, when the number of rounds 
    \[
    T \geq \log\frac{2k^2}{\alpha}\cdot 2\left(\frac{k\pmax}{r}\right)^2\cdot \left(\frac{1}{\pimin} + 1\right)^2,
    \]
    we have
    \begin{enumerate}
        \item if the seller's true pessimistic regret $\WstEr(c_0,\allocseq{x}) \leq r$, she passes with probability at least $1 - \alpha$; and
        \item if her plausible pessimistic regret $\WstEr(\plcost, \allocseq{x}) \geq 2r$, then she fails with probability at least $1 - \alpha$.
    \end{enumerate}
\end{theorem}

A direct corollary of the above theorem (together with \Cref{lem:indistinguishable}) is that the regret estimator in our refined auditing method outputs the least upper bound of the regret of the seller, because \Cref{lem:indistinguishable} asks the regret estimator to output approximately at least the pessimistic regret, and the following corollary says it outputs approximately at most the pessimistic regret. This corollary implies that our refined auditing method provides the most information about the transcript, given the one-sided consistency requirement. In other words, our method is at least as permissive as any auditing method that has a high probability of failing algorithmic outcomes with non-vanishing calibrated regret. 
\begin{corollary}
    Let the regret estimator in our refined auditing method be $\mathcal A$. It satisfies
    \[ \lim_{T \to \infty} \PPr{p^t \sim \pi^t}{\mathcal A(\Transc^T) > \WstEr(\allocseq{x}, c) + \eps} = 0, \]
    for any $\eps>0$, cost $c$, transcript $\Transc^T$, and sequence of ground-truth sequence of allocations $\{\allocseq{x}\}_{T \geq 1}$.
\end{corollary}
\begin{proof}
    A direct corollary from the proof of \Cref{thm:sample_complexity}.
\end{proof}

\paragraph{An extension: Auditing with aggregated price distributions.}
Recall that 
one of the motivations for designing the refined auditing method is to allow testing aggregated empirical distributions when price distribution data is not available. We end this section by showing that as long as the rate of change in price distributions is low, it is possible to run the same auditing method with the aggregated price distributions as a plug-in estimator for the true distributions. We do this by first showing a lemma that the aggregated price distributions are good approximators to the true price distributions, and then providing a sample complexity guarantee.
\begin{lemma}\label{lem:aggregate_possible}
    Consider $T$ rounds and $k$ price levels bounded in $[0, 1]$. In round $t$ the seller posts price $p^t \sim \pi^t$. Suppose $\|\pi^t-\pi^{t+1}\|_\infty \leq \eps$ for all $1 \leq t \leq T-1$.\footnote{For example, the multiplicative weights update (MWU) algorithm with learning rate $\eps$ satisfies the condition. See the ending remark of the proof of this lemma for details. Designing a test for such a condition is left as an open question.} Then there exists an algorithm that uses only price samples $\{p^t\}_{t=1}^T$ and outputs estimated price distributions $\{\hat{\pi}^t\}_{t=1}^T$ such that, with probability at least $1 - \delta$
    \[ \|\hat{\pi}^t-\pi^t\|_\infty \leq \left(4 \eps \log \frac{2Tk}{\delta} \right)^{1/3}, \]
    for all $1 \leq t \leq T$.
\end{lemma}
\begin{corollary}\label{cor:aggregate_sample_complexity}
    Consider the environment with maximum possible price level $\pmax$ and $k=|\priceset|$ price levels. Let $c_0$ be the seller's true cost and $
    \plcost = \arg\min_{c\in\costrange}\WstEr(c,\allocseq{x})$ be the plausible cost of the seller. Fix confidence level $1-\delta$, threshold $r$. Suppose the seller's pricing algorithm has the property that, for all $T'$, the price distributions $\{\pi^t\}_{t=1}^T$ generated by running the algorithm for $T'$ rounds satisfy
    \begin{enumerate}
        \item $\|\pi^t-\pi^{t+1}\|_\infty \leq (T')^{-\gamma}$ for all $1 \leq t \leq T'-1$ and some $\gamma>0$, and
        \item for some $\pimin$, for all $p \in C^t$ and $1 \leq t \leq T$, $\pi^t(p) \geq \pimin$, i.e., the minimum probability in the support of price distributions is at least $\pimin$,
    \end{enumerate}
   then there exists an auditing method without access to $\{\pi^t\}_{t=1}^T$, when the number of rounds 
    \[ T \geq \max \left\{t_0, \left(\frac{4(8\pmax k +r\pimin)^3}{r^3\pimin^6}\right)^{2/\gamma}, \frac{16k^2}{r^2}\log \frac{8k^2}{\delta}\cdot\left(\frac{1}{\pimin}+1\right)^2\cdot \pmax^2, \left(\frac{4}{\pimin^3} \right)^{2/\gamma} \right\}+1, \]
where $t_0 = \sup \left\{t \in \mathbb R:t^\frac{\gamma}{2} \leq \log \frac{8tk^3}{\delta} \right\}$,
    for which
    \begin{enumerate}
        \item if the seller's true pessimistic regret $\WstEr(c_0,\allocseq{x}) \leq r$, she passes with probability at least $1 - \delta$; and
        \item if her plausible pessimistic regret $\WstEr(\plcost, \allocseq{x}) \geq 2r$, then she fails with probability at least $1 - \delta$.
    \end{enumerate}
\end{corollary}

In the following two sections, we study two technical details on the concepts and assumptions used in our auditing, which have practical implications in law. In \Cref{sec:external_manipulable} we provide justifications that the stronger calibrated regret must be used instead of weaker best-in-hindsight regret by arguing that best-in-hindsight regret fails to exclude supra-competitive algorithms even when useful side information is guaranteed not to exist. In \Cref{sec:unknown_costs}, we demonstrate that it is possible for algorithms to pass the audit by pretending to have higher costs than they actually do. For such scenarios, the rule of reason can be applied to bound the range of costs to those that are reasonable for the domain.

\section{Best-in-hindsight Regret is Manipulable}\label{sec:external_manipulable}
Much online learning literature develops algorithms to satisfy vanishing best-in-hindsight (a.k.a. external) regret. But \cite{hartline2024regulation} argue that the stronger vanishing calibrated regret is essential for the regulation of (non-)collusion. They do this with an example \emph{with side information} where a seller posts supra-competitive prices while having non-positive best-in-hindsight regret.

Intuitively, regulation using calibrated regret is necessary with side information because we need to test whether the seller best responds to such useful information. A natural question is: If side information is guaranteed not to exist, can we use the weaker notion of (vanishing) best-in-hindsight regret? We give a negative answer in this section. We show a stronger argument by demonstrating that even in environments without side information, algorithms can have vanishing best-in-hindsight regret while being \emph{manipulable} (to be defined shortly) to supra-competitive prices. Therefore, if we use the more permissive vanishing best-in-hindsight regret as the definition of non-collusion, then there could be more anti-competitive algorithms passing the audit  (recall \nameref{enum:less_bad}), even if side information is guaranteed not to exist.

Combining the above argument with the fact that (simultaneous) calibrated regret minimization leads to approximate correlated equilibria \citep{foster1997calibrated},
we conclude that it is important to require vanishing calibrated regret in the definition of non-collusion even when there is no side information.

We first define what we mean by manipulation, and then construct the manipulation example.
\begin{definition}[Manipulation]
    Fix an instance of imperfect price competition with two sellers and a property $\mathcal A$ of pricing algorithms. Suppose seller 2's algorithm satisfies property $\mathcal A$. A \emph{manipulation of $\mathcal A$-algorithm} is an algorithm of seller 1 such that, when played against seller 2's algorithm, yields an outcome with the following properties:
    \begin{enumerate}
        \item (supra-competitive) \emph{Both} sellers have payoffs \emph{strictly higher} than any correlated equilibrium.
        \item (undetectable) \emph{Both} seller 1 and seller 2's algorithms satisfy property $\mathcal A$.
    \end{enumerate}
\end{definition}
Next, we show a manipulation of no-best-in-hindsight-regret algorithms (i.e.\;$\mathcal A =$ no best-in-hindsight regret).
The manipulation implies that there exists a scenario where the non-collusion definition of vanishing best-in-hindsight regret fails to identify a supra-competitive scenario.

Consider a setting of dynamic imperfect price competitions with 2 sellers. Let $V_1, V_2$ be the highest (correlated) equilibrium payoff for seller 1 and seller 2, in which they play a joint distribution of prices $\pi^e$. If the equilibrium strategy is pure, then let $p_1^e, p_2^e$ be the prices they play. We consider a family of commonly used best-in-hindsight regret minimization algorithms defined as follows.
\begin{definition}[$\gamma$-mean-based learning, \cite{braverman2018selling}]
    Fix horizon $T$ and $\gamma=o(1)$. Let $\sigma_{p, t} = \sum_{s=1}^tu^s(p)$ be the cumulative utilities for posting price $p$ in the first $t$ rounds. A seller is \emph{$\gamma$-mean-based} if the seller posts price $p$ w.p.\;at most $\gamma$ as long as there exists another price $q$ such that $\sigma_{q, t}>\sigma_{p, t}+\gamma T$.
\end{definition}
Many vanishing best-in-hindsight-regret algorithms (e.g., exponential-weight algorithm for exploration and exploitation, follow the perturbed leader) are known to be $\gamma$-mean-based learning algorithms \citep{braverman2018selling}. In the following theorem, we show that a seller running a $\gamma$-mean-based learning algorithm is suspectible to manipulation into supra-competitive outcomes. The manipulator can also achieve no best-in-hindsight when doing such a manipulation.
\begin{theorem}\label{thm:manipulate}
    There exists an instance in which the environment is stationary across rounds, both sellers have no side information%
    , and seller 1 can achieve an outcome with the following properties against seller 2 who is using any $\gamma$-mean-based learning algorithm:
    \begin{enumerate}
        \item (supra-competitive) for $\Omega(T)$ rounds, both play $p_1>p_1^e, p_2>p_2^e$ in each round with high probability,\footnote{In our construction, the maximum-payoff correlated equilibrium is a pure equilibrium.}
        \item (no loss of payoff) receive expected payoff $V_1'T-o(T), V_2'T-o(T)$ where constants $V_1'>V_1, V_2'>V_2$, and
        \item (no best-in-hindsight regret) both seller 1 and seller 2 have vanishing best-in-hindsight regret.
    \end{enumerate}
\end{theorem}
In words, \Cref{thm:manipulate} says that seller 2 can be manipulated into a significant number (constant fraction) of rounds with supra-competitive prices, while both sellers have no best-in-hindsight regret and get higher-than-equilibrium expected payoffs.

\section{On the Effect of Unknown Costs}\label{sec:unknown_costs}
In this section, we demonstrate that imprecision in cost information might significantly impact the efficacy of auditing from pricing data and present challenges to regulating collusion.

Recall that in the auditing problem, we assume that the cost of the seller is unknown to the regulator and the seller passes the audit as long as the \emph{plausible} regret is low (recall that the plausible regret is the minimum regret that a seller with cost in $\costrange$ can have).
Therefore, it is conceivable that algorithms could find a supra-competitive outcome for the true costs but pass the audit for different costs.
We ask whether such a phenomenon can actually occur.  More precisely, our question is

\begin{quote}
    Are there natural algorithms that when configured with the \emph{true} cost $c$, find outcomes that are considered competitive by auditing methods for a higher cost $c'$, while actually being supra-competitive for cost $c$?
\end{quote}

We find an affirmative answer to this question with two implications.\;First,\;the\;imprecise knowledge of the cost of the seller could dramatically affect the result of the audit. Even when configured with true costs, supra-competitive algorithms might still pass the test because the outcome is consistent with low regret for some other plausible costs.
Second, such behavior is hard to distinguish from genuine competition by only looking at the pricing data, which is a new challenge for auditing algorithmic collusion.

A mitigation to this problem is applying rule of reason: A further investigation of the sellers and market contexts is needed to narrow down the range of reasonable costs.

\begin{remark}
    To clarify the formulation of the question,  note that the following two scenarios are different:
\begin{enumerate}
    \item A seller deliberately inputs a fake cost to the pricing algorithm, causing the algorithm to post prices higher than what is optimal for her true cost.\label{enum:human_collusion}
    \item A seller truthfully reports her cost to the pricing algorithm, but the algorithm finds a supra-competitive outcome that looks competitive with a higher cost.\label{enum:algo_collusion}
\end{enumerate}
Since algorithmic collusion refers to collusion facilitated by the algorithm, we consider only scenario \ref{enum:algo_collusion} as algorithmic collusion.
\end{remark}

\paragraph{Details of the simulation and results.}
We demonstrate by simulation the possibility that algorithms can find a supra-competitive outcome that is seen as competitive for a high inferred plausible cost.\footnote{Code for replicating the simulation can be found at \url{https://github.com/wangchang327/collusion-cslaw25}.} The instance of dynamic imperfect price competition is the same as that in \citet{hartline2024regulation} and the configuration of the algorithm resembles that in \citet{banchio2022artificial}. We consider two sellers with costs $c_1=0.1$ and $c_2=0.2$, respectively. The grid of allowable price levels are from $0.05$ to $0.95$ with step size $0.05$. In each round, the buyer's valuations of the two sellers' goods are i.i.d.\;uniformly distributed over $[0, 1] \times [0, 1]$. The sellers post prices and the reward of the seller is the expected payment from the buyer, net her own cost. At the end of each round, each seller records her posted price, the demand, and her price distribution.
The competition lasts for $T=10^6$ rounds and the experiment with the same setup is repeated 100 times. Sellers compete with each other using the same simple stateless Q-learning algorithm as that in \citet{banchio2022artificial}, which only keeps track of the estimated continuation payoff for each price level $p \in \priceset$ in the Q-table. The hyper-parameters are chosen as follows. The Q-learning uses an $\varepsilon$-greedy strategy with optimistic initialization and exploration probability $\varepsilon=0.01$.  The Q-table is updated according to the standard rule
\[ Q^{t+1}(p)=(1 - \alpha)Q^{t+1}(p) + \alpha(u^t(p) + \gamma \max_{q \in \priceset} Q^t(q)) \quad \forall p \in \priceset, \]
where the learning rate $\alpha=0.05$ and discount factor $\gamma=0.99$. Note that this stateless Q-learning in competition is a mis-specified model.\;Thus the Q-table payoffs are statistically inaccurate and this inaccuracy is the channel by which the algorithms achieve supra-competitive prices.

In \Cref{fig:collusive_heatmap} we confirm that Q-learning exhibits supra-competitive behavior. Q-learning converges to strategies that both sellers post prices $(0.6, 0.65)$ in most cases (with payoffs 0.169 and 0.122).\;The $(0.6, 0.65)$ outcome is supra-competitive as its prices exceed the highest payoff correlated equilibrium of the game.\;The highest payoff equilibrium of the game is that the sellers post $0.5$ and $0.55$ (with payoffs 0.159 and 0.114), respectively. Note from the values of the payoffs that in this example, the payoff difference between the sellers' supra-competitive outcome and the equilibrium is small, thus the regret of the sellers is expected to be at a relatively small scale $10^{-2} \sim 10^{-3}$.

\begin{figure}[htb]
    \centering
    \includegraphics[width=0.6\textwidth]{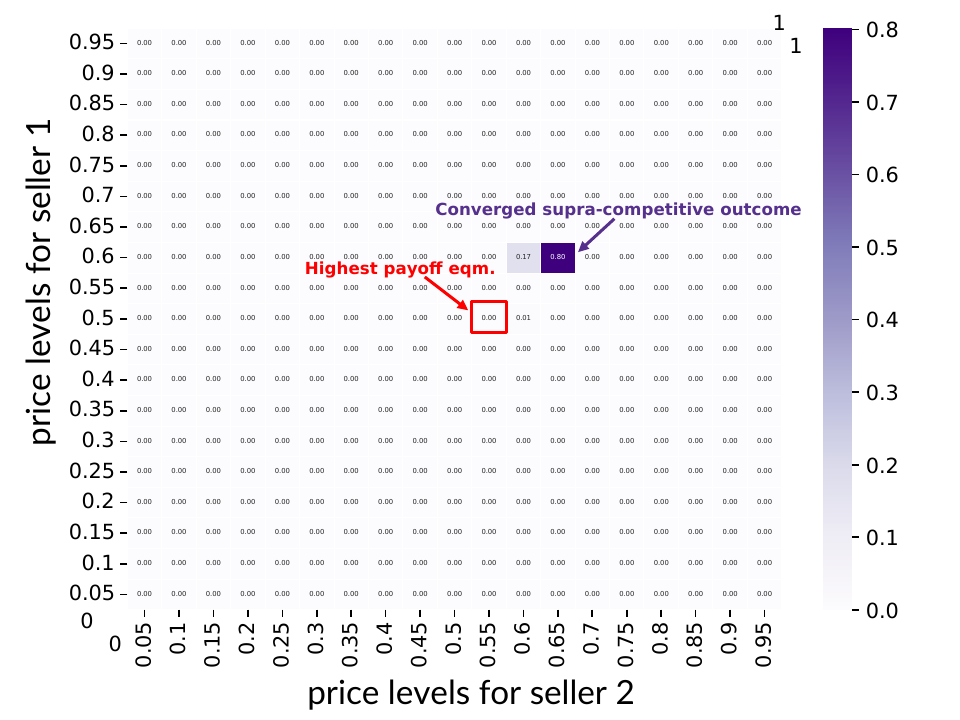}
     \caption{The frequencies of each pair of strategies in the last 10 rounds of the competition are shown in the heatmap. The highest payoff equilibrium is marked in red.}
     \label{fig:collusive_heatmap}
\end{figure}

After the transcripts are generated, we audit the transcripts.\;As a baseline, we use the prices posted by the opponent to compute the true expected regret of the seller.\;We then audit the transcripts using our auditing method.\;The auditing result is shown in \Cref{fig:audit_result_without_decay}.

\begin{figure}[htb]
    \centering
    \begin{subfigure}{0.49\textwidth}
        \includegraphics[width=\textwidth]{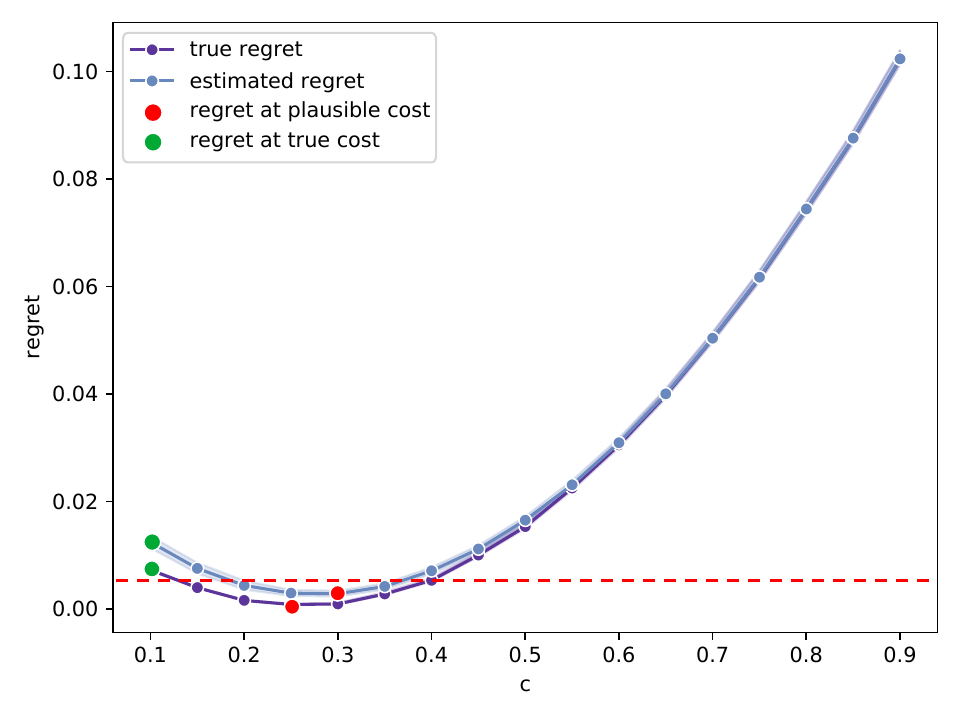}
        \caption{Range of costs $[0.1, 0.9]$}
    \end{subfigure}
    \hfill
    \begin{subfigure}{0.49\textwidth}
        \includegraphics[width=\textwidth]{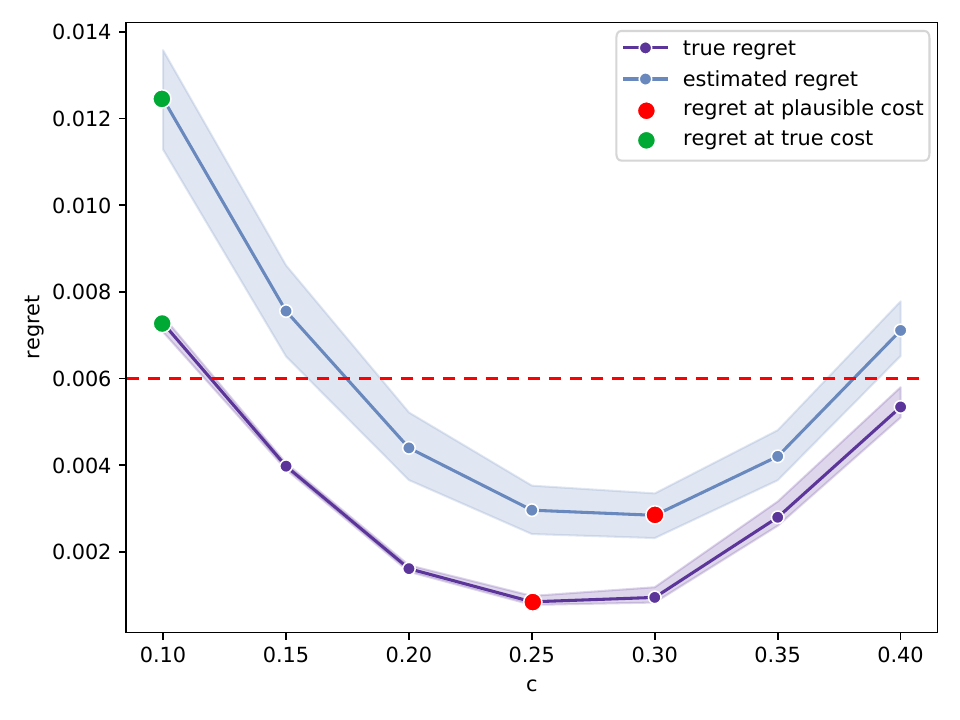}
        \caption{Range of costs $[0.1, 0.4]$}
    \end{subfigure}
     \caption{The true regret and estimated regret plotted against different assumed costs of seller 1, with a zoomed-view on the right. The red dashed line shows the example threshold $6\times10^{-2}$ for estimated regret. The seller maintains exploration. The true regret is the regret estimated with the knowledge of the ground truth demand, while the estimated regret is the regret estimated using our auditing method.}
     \label{fig:audit_result_without_decay}
\end{figure}
\footnotetext{That of seller 2 is similar.}

\begin{figure}[htb]
    \centering
    \includegraphics[width=0.6\linewidth]{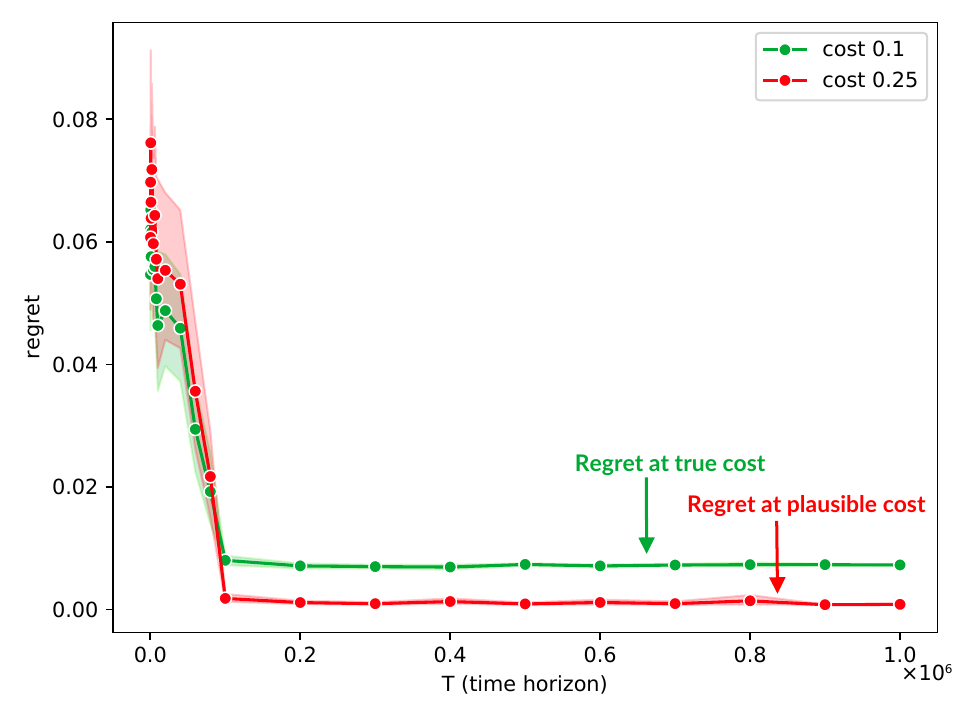}
    \caption{The true regret at the true cost $c=0.1$ and at the plausible cost $c'=0.25$ plotted with different time horizons.}
    \label{fig:regret_t}
\end{figure}

In \Cref{fig:audit_result_without_decay} we observe that the plausible estimated regret of the seller is around $3 \times 10^{-3}$ at $c'=0.3$.\footnote{Recall that in this example, the payoff difference between the sellers' supra-competitive outcome and the equilibrium is small, thus the regret of the sellers is at a relatively small scale $10^{-2} \sim 10^{-3}$.} The true cost of the seller is in fact $c=0.1$, for which the estimated regret is around $12 \times 10^{-3}$. If we set the auditing threshold to be, say, $6 \times 10^{-3}$, then the algorithm should be flagged as supra-competitive as the regret at $c=0.1$ is more than $6 \times 10^{-3}$. But since the plausible regret is approximately $3 \times 10^{-3}<6 \times 10^{-3}$, the seller will not be flagged. A similar pattern also applies to the true regret curve (i.e.,\;even when the regret estimation is exact) at $c'=0.25$ .
In \Cref{fig:regret_t}, we also plot the true regret at the true cost $c=0.1$ and at the plausible cost $c'=0.25$ with different time horizons. We observe that
both the regret at the true cost and the plausible cost converge at $T=10^6$.
The regret at the true cost exceeds the regret at the plausible cost by a significant margin. Excluding the loss due to exploration, the regret at the true cost is not vanishing and the regret at the plausible cost is vanishing.

From the observations above, we conclude that although we know the Q-learning algorithm converges to supra-competitive prices, if we assume that the viable costs of the seller include $c'=0.3$, then the Q-learning algorithm turns out to have a low plausible calibrated regret. In other words, when configured with a low cost (such as $c=0.1$), the Q-learning algorithm finds outcomes that look competitive for a higher cost $c'$ (such as $c'=0.3$). The results of this experiment confirm that it is possible for algorithms to find a supra-competitive outcome that is seen as competitive at a high inferred plausible cost.

\section{Conclusion}
In this work, we explore several questions around auditing (non-)collusion for pricing algorithms from data based on the framework of \citet{hartline2024regulation}. We motivate and interpret our study under the legal doctrines of antitrust analysis. %
We develop a refined auditing method that relaxes the previous requirement that a pricing algorithm must use fully-supported price distributions to be auditable by testing the pessimistic regret, thus allowing more efficiency-enhancing algorithms to be auditable.
We give an example demonstrating that requiring vanishing-calibrated regret as the non-collusion definition is essential to eliminate more collusion-promoting algorithms even without side information. Our experiment results show that under the current auditing framework, a regulator with very limited knowledge about a seller's cost may be unable to detect supra-competitive behavior of the seller, which suggest a rule of reason can be useful in antitrust analysis.
Open questions include: (a) designing a test for small learning rates to remove the need to record price distributions and (b) improving the sample complexity bound.

\bibliographystyle{plainnat}
\bibliography{references}

\begin{thebibliography}{29}
\providecommand{\natexlab}[1]{#1}
\providecommand{\url}[1]{\texttt{#1}}
\expandafter\ifx\csname urlstyle\endcsname\relax
  \providecommand{\doi}[1]{doi: #1}\else
  \providecommand{\doi}{doi: \begingroup \urlstyle{rm}\Url}\fi

\bibitem[Arunachaleswaran et~al.(2024)Arunachaleswaran, Collina, Kannan, Roth, and Ziani]{arunachaleswaran2024algorithmic}
Eshwar~Ram Arunachaleswaran, Natalie Collina, Sampath Kannan, Aaron Roth, and Juba Ziani.
\newblock Algorithmic collusion without threats.
\newblock \emph{arXiv preprint arXiv:2409.03956}, 2024.

\bibitem[Asker et~al.(2022)Asker, Fershtman, and Pakes]{asker2022artificial}
John Asker, Chaim Fershtman, and Ariel Pakes.
\newblock Artificial intelligence, algorithm design, and pricing.
\newblock In \emph{AEA Papers and Proceedings}, volume 112, pages 452--456. American Economic Association 2014 Broadway, Suite 305, Nashville, TN 37203, 2022.

\bibitem[Asker et~al.(2023)Asker, Fershtman, and Pakes]{asker2023impact}
John Asker, Chaim Fershtman, and Ariel Pakes.
\newblock The impact of artificial intelligence design on pricing.
\newblock \emph{Journal of Economics \& Management Strategy}, 2023.

\bibitem[Banchio and Mantegazza(2023)]{banchio2023adaptive}
Martino Banchio and Giacomo Mantegazza.
\newblock Adaptive algorithms and collusion via coupling.
\newblock In \emph{Proceedings of the 24th ACM Conference on Economics and Computation}, EC '23, page 208, New York, NY, USA, 2023. Association for Computing Machinery.
\newblock ISBN 9798400701047.
\newblock \doi{10.1145/3580507.3597726}.
\newblock URL \url{https://doi.org/10.1145/3580507.3597726}.

\bibitem[Banchio and Skrzypacz(2022)]{banchio2022artificial}
Martino Banchio and Andrzej Skrzypacz.
\newblock Artificial intelligence and auction design.
\newblock In \emph{Proceedings of the 23rd ACM Conference on Economics and Computation}, pages 30--31, 2022.

\bibitem[Braverman et~al.(2018)Braverman, Mao, Schneider, and Weinberg]{braverman2018selling}
Mark Braverman, Jieming Mao, Jon Schneider, and Matt Weinberg.
\newblock Selling to a no-regret buyer.
\newblock In \emph{Proceedings of the 2018 ACM Conference on Economics and Computation}, pages 523--538, 2018.

\bibitem[Calvano et~al.(2020)Calvano, Calzolari, Denicolo, and Pastorello]{calvano2020artificial}
Emilio Calvano, Giacomo Calzolari, Vincenzo Denicolo, and Sergio Pastorello.
\newblock Artificial intelligence, algorithmic pricing, and collusion.
\newblock \emph{American Economic Review}, 110\penalty0 (10):\penalty0 3267--97, 2020.

\bibitem[Chassang and Ortner(2023)]{chassangRegulatingCollusion2023}
Sylvain Chassang and Juan Ortner.
\newblock Regulating {{Collusion}}.
\newblock \emph{Annual Review of Economics}, 15\penalty0 (1):\penalty0 177--204, 2023.
\newblock \doi{10.1146/annurev-economics-051520-021936}.
\newblock URL \url{https://doi.org/10.1146/annurev-economics-051520-021936}.

\bibitem[Chassang et~al.(2022)Chassang, Kawai, Nakabayashi, and Ortner]{chassang2022robust}
Sylvain Chassang, Kei Kawai, Jun Nakabayashi, and Juan Ortner.
\newblock Robust screens for noncompetitive bidding in procurement auctions.
\newblock \emph{Econometrica}, 90\penalty0 (1):\penalty0 315--346, 2022.

\bibitem[{City and County of San Francisco}(2024)]{sfban2024}
{City and County of San Francisco}.
\newblock Administrative code - ban on automated rent-setting.
\newblock \url{https://sfgov.legistar.com/LegislationDetail.aspx?ID=6789588&GUID=89BA28F7-B3B8-44D0-806B-FFDC5FC29015}, 2024.
\newblock Accessed: 2024-09-27.

\bibitem[Deng et~al.(2019)Deng, Schneider, and Sivan]{deng2019strategizing}
Yuan Deng, Jon Schneider, and Balasubramanian Sivan.
\newblock Strategizing against no-regret learners.
\newblock \emph{Advances in neural information processing systems}, 32, 2019.

\bibitem[{Department of Justice}(2024)]{dojrealpage2024}
{Department of Justice}.
\newblock Justice department sues realpage for algorithmic pricing scheme that harms millions of american renters.
\newblock \url{https://www.justice.gov/opa/pr/justice-department-sues-realpage-algorithmic-pricing-scheme-harms-millions-american-renters}, 2024.
\newblock Accessed: 2024-08-31.

\bibitem[\emph{Addyston Pipe \& Steel Co.}~v. \emph{United States}(1899)]{1899addyston}
\emph{Addyston Pipe \& Steel Co.}~v. \emph{United States}, 1899.

\bibitem[\emph{Brooke Group Ltd.}~v. \emph{Brown \& Williamson Tobacco Corp.}(1993)]{1993brooke}
\emph{Brooke Group Ltd.}~v. \emph{Brown \& Williamson Tobacco Corp.}, 1993.

\bibitem[\emph{Chicago Board of Trade}~v. \emph{United States}(1918)]{1918chicago}
\emph{Chicago Board of Trade}~v. \emph{United States}, 1918.

\bibitem[\emph{Standard Oil Co. of New Jersey}~v. \emph{United States}(1911)]{1911standard}
\emph{Standard Oil Co. of New Jersey}~v. \emph{United States}, 1911.

\bibitem[Fish et~al.(2024)Fish, Gonczarowski, and Shorrer]{fish2024algorithmic}
Sara Fish, Yannai~A Gonczarowski, and Ran~I Shorrer.
\newblock Algorithmic collusion by large language models.
\newblock \emph{arXiv preprint arXiv:2404.00806}, 2024.

\bibitem[Foster and Vohra(1997)]{foster1997calibrated}
Dean~P Foster and Rakesh~V Vohra.
\newblock Calibrated learning and correlated equilibrium.
\newblock \emph{Games and Economic Behavior}, 21\penalty0 (1-2):\penalty0 40, 1997.

\bibitem[Gavil(2011)]{gavil2011moving}
Andrew~I Gavil.
\newblock Moving beyond caricature and characterization: The modern rule of reason in practice.
\newblock \emph{S. Cal. L. Rev.}, 85:\penalty0 733, 2011.

\bibitem[Hansen et~al.(2021)Hansen, Misra, and Pai]{hansen2021frontiers}
Karsten~T Hansen, Kanishka Misra, and Mallesh~M Pai.
\newblock Frontiers: Algorithmic collusion: Supra-competitive prices via independent algorithms.
\newblock \emph{Marketing Science}, 40\penalty0 (1):\penalty0 1--12, 2021.

\bibitem[Harrington(2018)]{harrington2018developing}
Joseph~E Harrington.
\newblock Developing competition law for collusion by autonomous artificial agents.
\newblock \emph{Journal of Competition Law \& Economics}, 14\penalty0 (3):\penalty0 331--363, 2018.

\bibitem[Harrington(2022)]{harrington2022effect}
Joseph~E Harrington.
\newblock The effect of outsourcing pricing algorithms on market competition.
\newblock \emph{Management Science}, 68\penalty0 (9):\penalty0 6889--6906, 2022.

\bibitem[Hartline et~al.(2024)Hartline, Long, and Zhang]{hartline2024regulation}
Jason~D Hartline, Sheng Long, and Chenhao Zhang.
\newblock Regulation of algorithmic collusion.
\newblock In \emph{Proceedings of the Symposium on Computer Science and Law}, pages 98--108, 2024.

\bibitem[Hovenkamp(2018)]{hovenkamp2018rule}
Herbert Hovenkamp.
\newblock The rule of reason.
\newblock \emph{Fla. L. Rev.}, 70:\penalty0 81, 2018.

\bibitem[{\emph{In re Text Messaging Antitrust Litigation}}(2015)]{InreTextMessagingAntitrustLitigation2015}
{\emph{In re Text Messaging Antitrust Litigation}}.
\newblock \emph{F. 3d}, 782\;\penalty0 (No. 14-2301):\penalty0 867, 2015.

\bibitem[Klein(2021)]{klein2021autonomous}
Timo Klein.
\newblock Autonomous algorithmic collusion: Q-learning under sequential pricing.
\newblock \emph{The RAND Journal of Economics}, 52\penalty0 (3):\penalty0 538--558, 2021.

\bibitem[Nekipelov et~al.(2015)Nekipelov, Syrgkanis, and Tardos]{Nekipelov2015}
Denis Nekipelov, Vasilis Syrgkanis, and \'Eva Tardos.
\newblock Econometrics for learning agents.
\newblock In \emph{Proceedings of the Sixteenth ACM Conference on Economics and Computation}, EC '15, page 1–18, New York, NY, USA, 2015. Association for Computing Machinery.
\newblock ISBN 9781450334105.
\newblock \doi{10.1145/2764468.2764522}.
\newblock URL \url{https://doi.org/10.1145/2764468.2764522}.

\bibitem[Sawyer(2019)]{sawyer2019us}
Laura~Phillips Sawyer.
\newblock \emph{US antitrust law and policy in historical perspective}.
\newblock Harvard Business School, 2019.

\bibitem[Watkins(1989)]{watkins1989learning}
Christopher John Cornish~Hellaby Watkins.
\newblock \emph{Learning from delayed rewards}.
\newblock PhD thesis, King's College, Cambridge United Kingdom, 1989.

\end{thebibliography}

\newpage
\section*{Appendices}
\appendix
\section{Regret over Full Randomness of the Algorithm Cannot be Estimated Consistently}\label{appendix:full_randomness}
As we have discussed in \Cref{def:swap_regret} and \Cref{rmk:full_randomness}, the regret we are estimating is not the true regret of the algorithm itself. It might be tempting to consider the latter which is a stronger notion and is the usual regret definition in the analysis of an algorithm. In this appendix, we show that it is not possible to estimate consistently even with fully-supported price distributions. Details in this appendix require the knowledge of definitions from \Cref{sec:regret_estimation}.

Observe that for online learning algorithms, the distributions $\distseq$ are potentially random since the sampling distribution at $t+1$, $\pi^{t+1}$, depends on the realization of the price $p^t$ sampled at $t$. This motivates the question of estimating the regret with respect to the full randomness of the algorithm. We first define an online posted-pricing algorithm as follows:

\begin{definition}[Online posted-pricing algorithm]
    Call $\Transc^t = \{x^s(p^s), p^s, \pi^s\}_{s=1}^t$ a \emph{transcript} of length $t$. The set of all the length-$t$ transcripts is denoted as $\Hist^t$. Let $\Hist^0 = \varnothing$. An \emph{online posted-pricing algorithm} is a mapping $\mathcal M:\bigcup_{t \geq 0} \ \Hist^t \to \Delta(\priceset)$. Given the ground-truth $\allocseq{x}$, a $T$-round realization of algorithm $\mathcal M$ is a transcript of length $T$ where
    \[ \pi^{t+1} = \mathcal M(\{x^s(p^s), p^s, \pi^s\}_{s=1}^t), \quad p^{t+1} \sim \pi^{t+1} \quad (0 \leq t \leq T). \]
\end{definition}

Note that given the ground-truth $\allocseq{x}$ and $T$, a posted-pricing algorithm can be identified as a joint distribution (that depends on the ground-truth allocation) over $\{p^t, \pi^t\}_{t=1}^T$. The distribution of $p^t$ conditioned on $\pi^t$ is $\pi^t$.

Using this identification, we can write the expected regret of algorithm $\mathcal M$ given the ground-truth allocation $\allocseq{x}$ for a seller with cost $c$ as
\[ 
  \TrueEr(\Mech, \allocseq{x}, c) = \max_{\sigma: \priceset \to \priceset} \EExp{\{p^t, \pi^t\}_{t=1}^T \sim \Mech(\allocseq{x})}{\frac{1}{T}\sum_{t=1}^Tu(\sigma(p^t),x^t)-u(p^t,x^t)}.
\]

Unfortunately, it is impossible to estimate this regret consistently, even if all the price distributions are fully-supported.

\begin{proposition}\label{prop:full_randomness_impossible}
    There does not exist a regret estimator such that
    \[ \lim_{T \to \infty} \PPr{\{p^t, \pi^t\}_{t=1}^T \sim \Mech}{|\Algo(\Transc^T) - \TrueEr(\Mech, \allocseq{x}, c)| \geq \eps} = 0, \]
    even if we restrict attention to algorithms with fully-supported price distributions.
\end{proposition}
The intuition of the above proposition is that by looking at the transcript, the regret estimator has limited information about how the price distributions are generated, so it cannot consistently make inferences about the regret over the randomness of the price distributions.
\begin{proof}
Suppose not. Pick two algorithms $\mathcal M, \mathcal M'$, cost $c$ and $\allocseq{x}$ such that $\TrueEr(\Mech, \allocseq{x}, c)=r_1, \TrueEr(\Mech', \allocseq{x}, c)=r_2$, and $r_1-r_2 = \delta>0$. Create a new algorithm $\overline{\mathcal M}$ as follows: Before the first round, the algorithm secretly flips a fair coin and secretly decides to run $\mathcal M$ or $\mathcal M'$ for infinite time according to the result of the flip. Note that the regret estimator sees an infinite sequence of transcript, generated by either $\mathcal M$ or $\mathcal M'$, w.p. $\frac 12$ each. Therefore,
\[ \lim_{T \to \infty} \PPr{\{p^t, \pi^t\}_{t=1}^T \sim \overline{\Mech}}{|\Algo(\Transc^T) - r_1| \leq \eps} = \lim_{T \to \infty} \PPr{\{p^t, \pi^t\}_{t=1}^T \sim \overline{\Mech}}{|\Algo(\Transc^T) - r_2| \leq \eps} = \frac 12. \]
But note that $r_2 < \TrueEr(\overline{\Mech}, \allocseq{x}, c)<r_1$. Set $\eps_0<\frac 14\min\{\TrueEr(\overline{\Mech})-r_1, r_2-\TrueEr(\overline{\Mech})\}$ we then have
\[ \lim_{T \to \infty} \PPr{\{p^t, \pi^t\}_{t=1}^T \sim \overline{\Mech}}{|\Algo(\Transc^T) - \TrueEr(\overline{\Mech}, \allocseq{x}, c)| \geq \eps_0}=0, \]
a contradiction.
\end{proof}

\section{Omitted Proofs}\label{sec:oproofs}

\subsection{Proof of \texorpdfstring{\Cref{prop:reduction}}{Proposition 3.2}}
\begin{proof}
    The estimation algorithm: Note that the regret of any transcript is bounded by $\pmax$. We discretize the interval with step size $\eps$ and do $\ell=\pmax/\eps$ audits with thresholds $r_1=\eps, r_2=2\eps, \dotsc, r_\ell=\pmax$ \emph{simultaneously} using $\mathcal A$. Each threshold auditing returns either $[0, r_i]$ or $[r_i+\eps, \pmax]$, indicating which interval the true regret is in. Call this interval $J_i$. Let $J = \bigcap_{i=1}^kJ_i$. If $J$ contains only a single point or $|J| \leq \eps$, then we output the single point or the midpoint of $J$. Otherwise, we draw a guess uniformly random from $[0, \pmax]$.
    
    Analysis: Consider the ``good event'' that all the $\pmax/\eps$ audits satisfy the property in the statement of \Cref{prop:reduction}. Suppose the true regret is in $[j \eps, (j+1) \eps]$. Then conditioned on the good event, for all audits with threshold $r \geq (j+1) \eps$ they output S, and for all audits with $r \leq (j-1) \eps$ they output G, and for the audit with threshold $r=j \eps$, its output can be arbitrary. Observe that in this case the intersection of the outputs of audits except that with threshold $j \eps$ is exactly $[j \eps, (j+1) \eps]$. The final intersection is either $\{j \eps\}, \{(j+1) \eps\}$, or $[j \eps, (j+1) \eps]$ depending on the output of the audit with threshold $j \eps$ (which we do not have guarantee with). But all three possible outputs are within $\eps$ error of the true regret. We conclude the proof by noting that union bound gives that the good event happens with probability at least $1 - \frac{\pmax f(\eps, T)}{\eps}$.
\end{proof}

\subsection{Proof of \texorpdfstring{\Cref{prop:previous_consistency}}{Proposition 3.5}}
\begin{proof}
From the proof of Lemma A.1 in \cite{hartline2024regulation}, we have for any $\eps >0 $
   \[
   P_{\geq \eps}^T = \PPr{\priceseq \sim \distseq}{|\Algo(\Transc^T) - \TrueEr(\allocseq{x}, c)| \geq \eps} \leq 2k^2\exp\left(-\frac{\eps^2}{2k^2\sum_{t=1}^T d^2}\right)
   \]
   where $k=|\priceset|$ and 
   \[
   d = \frac{1}{T}\left(\frac{1}{\underline{\pi}^T}+1\right)\pmax.
   \]
   By assumption on the transcript, we have
   \[
   P_{\geq \eps}^T \leq 2k^2\exp\left(-\frac{\eps^2T}{2k^2\pmax^2 (\underline{\pi}^T+1)^2}\right) = o(T).
   \]
   Therefore, we have $\lim_{T\to \infty}P_{\geq \eps}^T = 0$.
\end{proof}

\subsection{Proof of \texorpdfstring{\Cref{prop:consistency_negative}}{Proposition 3.6}}
\begin{proof}
    Fix any regret estimator $\mathcal A$. Assume for contradiction that two-sided consistency holds. 
    Sort the prices in $\priceset$ as $p_1<p_2 < \dotsb < p_k$.
    
    Consider the following example. For all $1 \leq t \leq T$ we have $\pi^t(p_k)=0$ and $\pi^t(p_{k-1}) \neq 0$. Pick an arbitrary positive constant $a \leq 1$ and let $\allocseq{x}$ be such that for all $1 \leq t \leq T$, $x^t(p_i)= a$ and $x^t(p_k)=0$ for $1 \leq i<k$. Let $\allocseq{z}$ be another sequence of allocations such that
    \[ z^t(p)=x^t(p) \quad \text{for} \quad p=p_1, \dotsc, p_{k-1} \quad \text{and} \quad z^t(p_k)=x^t(p_{k-1})=a \]
    for all $1 \leq t \leq T$.

    By the assumption that two-sided consistency always holds, we have
    \[ \mathcal A(\{x^t(p^t), p^t, \pi^t\}_{t=1}^T) \overset{P}{\to}\TrueEr(c, \allocseq{x}), \quad \mathcal A(\{z^t(p^t), p^t, \pi^t\}_{t=1}^T) \overset{P}{\to}\TrueEr(c, \allocseq{z}). \]

    We claim that by construction the random variables $\mathcal A(\{x^t(p^t), p^t, \pi^t\}_{t=1}^T)$ and $\mathcal A(\{z^t(p^t), p^t, \pi^t\}_{t=1}^T)$ are equal w.p. 1. Therefore $\TrueEr(c, \allocseq{x}) = \TrueEr(c, \allocseq{z})$ holds. In fact,
    \begin{align*}
        \Pr_{p_t \sim \pi_t}[\{x^t(p^t), p^t, \pi^t\}_{t=1}^T \neq \{z^t(p^t), p^t, \pi^t\}_{t=1}^T] &\leq \sum_{t=1}^T \Pr_{p^t \sim \pi^t}[z^t(p^t) \neq x^t(p^t)]
        \\&= \sum_{t=1}^T \Pr[p^t=p_k] = \sum_{t=1}^T \pi^t(p_k)=0,
    \end{align*}
    which means the transcripts are the same w.p. 1, so any deterministic algorithm's outputs are the same w.p. 1.

    Next, we aim to show that $\TrueEr(c, \allocseq{z}) > \TrueEr(c, \allocseq{x})$, so we get a contradiction. Since $\pi^t(p_k)=0$ and $x^t(p)=z^t(p)$ for all $p \neq p_k$ and $1 \leq t \leq T$ we have
    \[ \frac 1T \sum_{t=1}^T \sum_p \pi^t(p)(p-c)x^t(p) = \frac 1T \sum_{t=1}^T \sum_p \pi^t(p)(p-c)z^t(p). \]
    So it suffices to show that
    \begin{equation}\label{eqn:proof_uneq_unbiased}
        \max_\sigma \frac 1T \sum_{t=1}^T \sum_p \pi^t(p) (\sigma(p)-c)z^t(\sigma(p)) > \max_\sigma \frac 1T \sum_{t=1}^T \sum_p \pi^t(p) (\sigma(p)-c)x^t(\sigma(p)).
    \end{equation}
    Note that the optimizer of the LHS is $\tau(p)=p_k$ for all $p \in \priceset$ and the optimizer of the RHS is $\rho(p)=p_{k-1}$. \Cref{eqn:proof_uneq_unbiased} follows since
    \[ \sum_{t=1}^T \sum_p \pi^t(p) (p_k-c)z^t(p_k) - \sum_{t=1}^T \sum_p \pi^t(p) (p_{k-1}-c)x^t(p_{k-1})
     = \sum_{t=1}^T \sum_p a \pi^t(p)(p_k-p_{k-1})>0. \]
     This completes the proof.
\end{proof}
\begin{remark}
    We note that this example further implies that there exists an algorithm with vanishing regret, but from its transcript, there is no two-sided consistent estimator for its regret. In fact, consider the following simple scenario: Let the opponent always play price $p_k$ and the ground truth allocation be $\allocseq{x}$.  Assume the seller being audited is best responding to her opponent. We cannot consistently (one-sided) estimate her regret according to the proof.
\end{remark}

\subsection{Proof of \texorpdfstring{\Cref{lem:indistinguishable}}{Proposition 3.9}}
We first observe the following lemma.
\begin{lemma}\label{lemma:indistinguishable_transcript}
    Fix the sequence of price distributions $\distseq$, for any $\allocseq{x} \sim _{\distseq} \allocseq{z}$ and deterministic regret estimator $\mathcal A$
    \[ \PPr{p^t \sim \pi^t}{\Algo(\{x^t(p^t), p^t, \pi^t\}_{t=1}^T) = \Algo(\{z^t(p^t), p^t, \pi^t\}_{t=1}^T)}=1. \]
\end{lemma}
\begin{proof}
    Similar to the proof of \Cref{prop:consistency_negative}
    \begin{align*}
        \Pr_{\{p^t, \pi^t\}_{t=1}^T \sim \Mech}[\{x^t(p^t), p^t, \pi^t\}_{t=1}^T &\neq \{z^t(p^t), p^t, \pi^t\}_{t=1}^T]
        \\&\leq \sum_{t=1}^T \Pr_{p^t \sim \pi^t}[z^t(p^t) \neq x^t(p^t)]
        \\&= \sum_{t=1}^T \Pr[p^t \notin C^t]=0. && (\text{because $\allocseq{x} \sim _{\distseq} \allocseq{z}$})
    \end{align*}
    So the transcripts are the same conditioned on $\distseq$ and the result follows.
\end{proof}

\begin{proof}[Proof of \Cref{lem:indistinguishable}]
    Let $\allocseq{x}_*$ be a sequence of allocations that achieves $\WstEr(c, \allocseq{x})$. By the one-sided consistency requirement
    \[ \lim_{T \to \infty} \Pr_{p^t \sim \pi^t}[\mathcal A(\{x_*^t(p^t), p^t, \pi^t\}_{t=1}^T) < \WstEr(c,\allocseq{x}) - \eps]=0. \]
    But \Cref{lemma:indistinguishable_transcript} implies that $\mathcal A(\{x^t(p^t), p^t, \pi^t\}_{t=1}^T) = A(\{x_*^t(p^t), p^t, \pi^t\}_{t=1}^T)$ w.p. 1, and the proposition follows.
\end{proof}

\subsection{Proof of \texorpdfstring{\Cref{thm:sample_complexity}}{Theorem 4.1}}
To prove the theorem, we present a few useful lemmas. We first characterize the location of the pessimistic allocations. We then show that the algorithm can consistently estimate the pessimistic regret.

\begin{lemma}
    Fix any $\distseq$. Let $[\allocseq{x}]$ be an equivalence class under the relation $\sim_{\distseq}$. Consider the following construction: Pick any $\allocseq{x} \in [\allocseq{x}]$ and set
    \[
    z_*^t(p) = \begin{cases}
        x^t(p) & (\text{if $p \in C^t$}), \\
        x^t(p')  & (\text{if } p' = \max\{r \leq p: r \in C^t\} \text{exists}), \\
        1  & (\text{otherwise}).
    \end{cases}
    \]

    Then $\allocseq{z}_*$ is well-defined, $\allocseq{z}_* \in [\allocseq{x}]$, and $\TrueEr(c, \allocseq{z}_*) = \sup_{\allocseq{z} \sim_{\distseq} \allocseq{x}} \TrueEr(c, \allocseq{z})$.
\end{lemma}
\begin{proof}
    Since any $\allocseq{x}$ in the equivalence class agrees on the prices that are in $C^t$ for all $1 \leq t \leq T$, and we also set $\allocseq{z}_*$'s allocation there the same, we have that $\allocseq{z}_*$ is well-defined and $\allocseq{z}_* \in [\allocseq{x}]$. To see that $\allocseq{z}_*$ achieves the supremum, note that
    \begin{multline}\label{eqn:proof_worst_case_allocation}
        \frac 1T \sum_{t=1}^T \sum_p \pi^t(p)\left[(\sigma(p)-c)z_*^t(\sigma(p))-(p-c)z_*^t(p)\right] \\ \geq \frac 1T \sum_{t=1}^T \sum_p \pi^t(p)\left[(\sigma(p)-c)x^t(\sigma(p))-(p-c)x^t(p)\right]
    \end{multline}
    for any $\allocseq{x} \in [\allocseq{x}]$ and any mapping $\sigma$. In fact, since $\allocseq{z}_* \sim_{\distseq} \allocseq{x}$ we have
    \[ \sum_{t=1}^T \sum_p \pi^t(p)(p-c)z_*^t(p)] = \sum_{t=1}^T \sum_p \pi^t(p)(p-c)x^t(p), \]
    and
    \[ \sum_{t=1}^T \sum_p \pi^t(p)(\sigma(p)-c)z_*^t(\sigma(p)) \geq \sum_{t=1}^T \sum_p \pi^t(p)(\sigma(p)-c)x^t(\sigma(p)) \]
    because $z_*^t(p) \geq x^t(p)$ for every $p \in \priceset$ and $1 \leq t \leq T$, by construction of $\allocseq{z}$.

    The lemma now follows from \Cref{eqn:proof_worst_case_allocation} and the fact that if $f(\sigma) \geq g(\sigma)$ everywhere, then $\max_\sigma f(\sigma) \geq \max_\sigma g(\sigma)$.
\end{proof}

\begin{lemma}\label{lem:deviation-bound-of-regret-estimator}
    Given a sequence of allocations $\allocseq{x}$. Let $k = |\priceset|$ be the number of price levels and $\pmax$ be the highest price. Given cost $c$, conditioned on observing the sequence of price distributions $\distseq$, for any fixed sequence of allocations $\distseq$, 
    \[
    \Pr[|\EstEr_{p, q}(c,\allocseq{x})-\WstEr_{p, q}(c,\allocseq{x})| \geq \eps] \leq 2\exp\left(-\frac{\eps^2}{2\sum_{t=1}^T d_t^2}\right)
    \]
    where 
    \[
     d^t = \frac{1}{T}\left(\frac{1}{\min_{p'\in C^t}\pi^t(p')}+1\right)\pmax.
    \]
\end{lemma}
We defer the proof of this lemma after presenting the proof of \Cref{thm:sample_complexity}.

Now we state the proof of \Cref{thm:sample_complexity}.
\begin{proof}
   \begin{enumerate}
       \item Starting from \Cref{lem:deviation-bound-of-regret-estimator}, we claim that for any fixed $c$
       \[
       \Pr[\EstEr(c, \allocseq{x}) - \WstEr(c, \allocseq{x}) \geq \eps] \leq 2k^2\exp\left(-\frac{\eps^2}{2k^2\sum_{t=1}^T d_t^2}\right).
       \]
       To bound the probability of
     \[
        \Pr\left[\EstEr(c, \allocseq{x}) - \WstEr(c, \allocseq{x}) \geq \eps\right] = \Pr\left[\sum_p \max_q \EstEr_{p, q}(c,\allocseq{x}) - \sum_p \max_q \WstEr_{p, q}(c,\allocseq{x})\geq \eps\right],
    \]
    note that
     \begin{align*}
     &\Pr\left[\sum_p \max_q \EstEr_{p, q}(c,\allocseq{x}) - \sum_p \max_q \WstEr_{p, q}(c,\allocseq{x}) \geq \eps\right] 
     \\\leq &\Pr\left[\exists p \in \priceset, \exists p'\in \priceset,\EstEr_{p,p'}(c,\allocseq{x}) - \WstEr_{p,p'}(c,\allocseq{x}) \geq \frac{\eps}{|\priceset|} \right]
     \\\leq & k^2\exp\left(-\frac{\eps^2}{2k^2\sum_{t=1}^T d_t^2}\right) &&      \mbox{(union bound)}.
    \end{align*}
       Plugging in the lower bound of $T$, $\eps=\delta^T$ and $c=c_0$, and $\delta^T$, we have 
       $\Pr[\EstEr(c_0, \allocseq{x}) - \WstEr(c_0,\allocseq{x}) \geq \delta^T] \leq \alpha$. By definition of $\estplcost$ and $c_0$, we have $\EstEr(\estplcost, \allocseq{x}) \leq \EstEr(c_0, \allocseq{x})$. 
        When the seller satisfies $\min_{p\in C^t,1\leq t \leq T}\pi^t(p) \geq \pimin$, we have $\delta^T \leq r/2$.
       Therefore, when $\WstEr(c_0,\allocseq{x}) \leq r$, 
       we have 

       \begin{align*}
       &\Pr[\EstEr(\estplcost, \allocseq{x}) +\delta^T \geq 2r] \\ \leq 
       & \Pr[\EstEr(\estplcost, \allocseq{x}) +\delta^T \geq r + 2\delta^T] \\
       \leq &\Pr[\EstEr(\estplcost, \allocseq{x}) +\delta^T \geq \WstEr(c_0,\allocseq{x}) + 2\delta^T]
       \\
           \leq
           & \Pr[\EstEr(\estplcost, \allocseq{x}) - \WstEr(c_0,\allocseq{x}) \geq \delta^T] \\
           \leq
           & \Pr[\EstEr(c_0, \allocseq{x}) - \WstEr(c_0,\allocseq{x}) \geq \delta^T] && \text{(since $\EstEr(\estplcost, \allocseq{x}) \leq \EstEr(c_0, \allocseq{x})$)} \\
           \leq &\alpha
       \end{align*}
       
       and the seller passes with 
       probability at least $1-\alpha$.
       \item Note that since $\estplcost$ is a random variable, we cannot use the same argument for fixed $c_0$ to bound the probability
       $\Pr[\EstEr(\estplcost, \allocseq{x}) - \WstEr(\plcost,\allocseq{x}) \leq -r]$ by plugging in $c=\estplcost$. Instead, observe that 
       \begin{align*}
       &\Pr\left[\EstEr(\estplcost, \allocseq{x}) - \WstEr(\estplcost,\allocseq{x}) \leq -r\right]
       \\\leq 
       &\Pr\left[\exists c,\EstEr(c, \allocseq{x}) - \WstEr(c,\allocseq{x}) \leq -r\right]
       \\=
       &\Pr\left[\exists c,\sum_p \max_q \EstEr_{p, q}(c,\allocseq{x}) - \sum_p \max_q \WstEr_{p, q}(c,\allocseq{x}) \leq -r\right]
       \\= 
       &\Pr\left[\exists c,\sum_p \max_q \WstEr_{p, q}(c,\allocseq{x})- \sum_p \max_q \EstEr_{p, q}(c,\allocseq{x}) \geq r\right]
       \\\leq
       &\Pr\left[\exists c,\exists p \in \priceset, \exists p'\in \priceset,\WstEr_{p,p'}(c,\allocseq{x}) - \EstEr_{p,p'}(c,\allocseq{x}) \geq \frac{r}{|\priceset|} \right]
       \\=
       &\Pr\left[\exists c,\exists p \in \priceset, \exists p'\in \priceset,\EstEr_{p,p'}(c,\allocseq{x}) - \WstEr_{p,p'}(c,\allocseq{x}) \leq -\frac{r}{|\priceset|} \right].
       \end{align*}
       Taking union bound over $p \in \priceset$ and $q' \in \priceset$, we have
        \begin{align*}   
        &\Pr\left[\exists c,\exists p \in \priceset, \exists q' \in \priceset, \EstEr_{p, p'}(c,\allocseq{x}) -  \WstEr_{p, p'}(c,\allocseq{x}) \leq -\frac{r}{|\priceset|}\right] 
        \\\leq
        & \sum_{p\in\priceset}\sum_{p'\in\priceset}\Pr\left[\exists c, \EstEr_{p, p'}(c,\allocseq{x}) -  \WstEr_{p, q'}(c,\allocseq{x}) \leq -\frac{r}{|\priceset|}\right].
       \end{align*}
       Observe that $\EstEr_{p, p'}(c,\allocseq{x}) -  \WstEr_{p, p'}(c,\allocseq{x})$ is linear in $c$ hence, when $c\in \costrange$,
       \begin{align*}
       &\Pr\left[\exists c, \EstEr_{p,p'}(c,\allocseq{x}) -  \WstEr_{p, p'}(c,\allocseq{x}) \leq -\frac{r}{|\priceset|}\right]\\\leq 
       &\Pr\left[\EstEr_{p, q'}(\ubar{c},\allocseq{x}) -  \WstEr_{p, q'}(\ubar{c},\allocseq{x}) \leq -\frac{r}{|\priceset|} \cup \EstEr_{p, q'}(\bar{c},\allocseq{x}) -  \WstEr_{p, q'}(\bar{c},\allocseq{x}) \leq -\frac{r}{|\priceset|}\right]
       \\\leq
       & \Pr\left[\EstEr_{p, q'}(\ubar{c},\allocseq{x}) -  \WstEr_{p, q'}(\ubar{c},\allocseq{x}) \leq -\frac{r}{|\priceset|}\right]+ \Pr\left[\EstEr_{p, q'}(\bar{c},\allocseq{x}) -  \WstEr_{p, q'}(\bar{c},\allocseq{x}) \leq -\frac{r}{|\priceset|}\right].
       \end{align*}
       \end{enumerate}
       Combining \Cref{lem:deviation-bound-of-regret-estimator}, we get
       \[
       \Pr\left[\EstEr(\estplcost, \allocseq{x}) - \WstEr(\estplcost,\allocseq{x}) \leq -r\right]
       \leq 2k^2\exp\left(-\frac{\eps^2}{2k^2\sum_{t=1}^T d_t^2}\right).
       \]
       Plugging in the bound of $T$ and $\eps = \delta^T$, we have 
       $\Pr\left[\EstEr(\estplcost, \allocseq{x}) - \WstEr(\estplcost,\allocseq{x}) \leq -\delta^T\right] \leq \alpha$. Therefore when $\WstEr(\plcost,\allocseq{x}) \geq 2r$, by definition of $\plcost$ and $\estplcost$, we have
       \begin{align*}&\Pr\left[\EstEr(\estplcost, \allocseq{x}) +\delta^T \leq 2r\right] \\ \leq  &\Pr\left[\EstEr(\estplcost, \allocseq{x}) - 2r \leq -\delta^T\right] \\
       \leq
    &\Pr\left[\EstEr(\estplcost, \allocseq{x}) - \WstEr(\plcost,\allocseq{x}) \leq -\delta^T\right] \\ 
    \leq
    &\Pr\left[\EstEr(\estplcost, \allocseq{x}) - \WstEr(\estplcost,\allocseq{x}) \leq -\delta^T\right] && \text{(since $\WstEr(\plcost, \allocseq{x}) \leq \WstEr(\estplcost, \allocseq{x})$)} \\ 
       \leq &\alpha.
    \end{align*}
       Hence the seller fails with probability at least $1-\alpha$.
\end{proof}

\subsubsection{Proof of \texorpdfstring{\Cref{lem:deviation-bound-of-regret-estimator}}{Lemma B.4}}
\begin{proof}
    Let 
    \[
    \widetilde{r}_{p,q}^t = \frac{1}{T}\left( \pi^t(p)\left[(q-c)\hat{x}^t(q)-(p-c)\hat{x}^t(p)\right]\right), \overline{r}_{p,q}^t = \frac{1}{T}\left( \pi^t(p)\left[(q-c)x_\ast^t(q)-(p-c)x_\ast^t(p)\right]\right),
    \]
    we claim that $\EExp{p^t \sim \pi^t}{\widetilde{r}_{p,q}^t} = \overline{r}_{p,q}^t$. In fact, first note that when $p \not \in C^t$, we have $\pi^t(p)=0$ and hence $\widetilde{r}_{p,q}^t=\overline{r}_{p,q}^t=0$ and the equality holds. Therefore, we assume from now that $\pi^t(p) \neq 0$. By definition of $\hat{x}^t$,  when $p' = \max\{r \leq p:r \in C^t\}$ exits, we have
    \begin{equation}
        \hat{x}^t(p) = 
        \begin{cases}
            x^t(p')/\pi^t(p') & p'=p^t, \\
                0 & \text{otherwise}.
        \end{cases},
    \end{equation} which implies that 
    $\EExp{p^t \sim \pi^t}{\hat{x}^t(p')} = x^t(p')$. By definition of $x_\ast^t$, we have $x^t(p') = x_\ast^t(p)$. 
     When $p'$ does not exist we have $\hat{x}^t(p)=x_\ast^t(p)=1$.
     Applying the same reasoning we also have $\EExp{p^t \sim \pi^t}{\hat{x}^t(q)} = x_\ast^t(q)$. Hence, by linearity of expectation we get $\EExp{p^t \sim \pi^t}{\widetilde{r}_{p,q}^t} = \overline{r}_{p,q}^t$. Since $p \leq \pmax, q\leq \pmax$, and $\hat{x}^t(p') \leq 1/\pi^t(p')$ for all $p' \in C^t$, we also have that
    \[
    |\widetilde{r}_{p,q}^t - r_{p,q}^t| \leq \frac{1}{T}\left(\frac{1}{\min_{p'\in C^t}\pi^t(p')}+1\right)\pmax.
    \]
    Note that 
    \[
    \EstEr_{p, q}(c,\allocseq{x}) - \WstEr_{p, q}(c) = \sum_{t=1}^T(\widetilde{r}_{p,q}^t - \overline{r}_{p,q}^t).
    \]
    Applying Azuma’s inequality, we get
    \[
    \Pr[\EstEr_{p, q}(c,\allocseq{x}) - \WstEr_{p, q}(c,\allocseq{x}) \geq \eps] \leq  k^2\exp\left(-\frac{\eps^2}{\sum_{t=1}^T d_t^2}\right).
    \]
    
    By a similar argument on the other side, we also have
    \[
     \Pr[\EstEr_{p, q}(c,\allocseq{x}) - \WstEr_{p, q}(c,\allocseq{x}) \leq -\eps] \leq k^2\exp\left(-\frac{\eps^2}{\sum_{t=1}^T d_t^2}\right).
    \]
    We get the desired result by combining two sides of the inequality.
\end{proof}

\subsection{Proof of \texorpdfstring{\Cref{lem:aggregate_possible}}{Proposition 4.3}}
\begin{proof}
    Let the price set be $\priceset$ where $|\priceset|=k$. We use the aggregation method as stated in the introduction. Let the window length be $L$. For each $\pi_i$, we use the empirical distribution of prices in the window centered at round $i$ to approximate $\pi_i$. Let the price distributions in the window be $F_1, \dotsc, F_L$ and $p_1 \sim F_1, \dotsc, p_L \sim F_L$. Then for any fixed $p \in \mathcal P$ and $j$ we have $\mathbb E[\mathbf 1_{\{p_j \leq p\}}]=F_j(p)$. Azuma--Hoeffding's inequality now gives
    \[ \Pr \left[\left|\frac 1L \sum_{j=1}^L(\mathbf 1_{\{p_j \leq p\}}-F_j(p)) \right| \geq t \right] \leq 2 \exp(-2Lt^2). \]
    With a union bound over all price levels, with probability at least $1-2k \exp(-2Lt^2)$ the aggregated empirical distribution $\tilde F$ satisfies $\|\tilde F - \frac 1L \sum_{j=1}^LF_j\|_\infty \leq t$. By the assumption of the lemma and the triangle inequality we have $\|F_i - \frac 1L \sum_{j=1}^LF_j\|_\infty \leq L \eps$ for all $i$. This implies $\|\tilde F - F_i\| \leq t+L \eps$ w.p. at least $1-2k \exp(-2Lt^2)$. With a union bound on all $T$ rounds, w.p. at least $1-2kT \exp(-2Lt^2)$, the aggregated estimator is ($t+L \eps$)-close to the true distribution in the $\ell_\infty$ distance in every round. Setting $\delta = 2kT \exp(-2Lt^2)$ gives the error bound
    \[ \frac{\log \frac{2Tk}{\delta}}{2t^2} \epsilon+t. \]
    The lemma follows by setting $\frac{\log \frac{2Tk}{\delta}}{2t^2} \epsilon=t$ and solving for the optimal $t$.
\end{proof}
\begin{remark}
    We provide an example application of the lemma. We show that the multiplicative weights update (MWU) algorithm with reward in $[0, 1]$ satisfies the condition. Let $V_1, \dotsc, V_k$ be the current cumulative reward for actions $1, \dotsc, k$. Now, the maximum change in the probability of playing action 1 occurs when the reward of action 1 is 1 and those for the other actions are 0. We bound the change as follows:
    \begin{align*}
        &\frac{(1+\eps)^{V_1+1}}{(1+\eps)^{V_1+1}+(1+\eps)^{V_2}+\dotsb+(1+\eps)^{V_k}} - \frac{(1+\eps)^{V_1}}{(1+\eps)^{V_1}+(1+\eps)^{V_2}+\dotsb+(1+\eps)^{V_k}}
        \\\leq&\frac{(1+\eps)^{V_1+1}}{(1+\eps)^{V_1+1}+(1+\eps)^{V_2}+\dotsb+(1+\eps)^{V_k}} - \frac{(1+\eps)^{V_1}}{(1+\eps)^{V_1+1}+(1+\eps)^{V_2}+\dotsb+(1+\eps)^{V_k}}
        \\=&\frac{(1+\eps)^{V_1+1}-(1+\eps)^{V_1}}{(1+\eps)^{V_1+1}+(1+\eps)^{V_2}+\dotsb+(1+\eps)^{V_k}}
        \leq\frac{(1+\eps)^{V_1+1}-(1+\eps)^{V_1}}{(1+\eps)^{V_1+1}+k-1}
        \\=&\frac{\eps}{1+\eps+\frac{k-1}{(1+\eps)^{V_1}}} \leq\frac{\eps}{\eps+1}\leq \eps.
    \end{align*}
    Therefore, the lemma can be used with MWU sellers (this includes a lot of MWU-style algorithms, such as EXP3).
\end{remark}

\subsection{Proof of \texorpdfstring{\Cref{cor:aggregate_sample_complexity}}{Corollary 4.4}}
The new algorithm is the same as the algorithm mentioned in \Cref{sec:general_algo} with the estimated distribution using $\Cref{lem:aggregate_possible}$, except that at the end the threshold $\delta^T$ to add is
\[ k\frac{\pmax \rho'}{\pimin}\left(\frac
{1}{\pimin-\rho'}+1\right) + \sqrt{\log\frac{8k^2}{\delta}\cdot 2 \left(\frac
{1}{\pimin}+1 \right)^2\pmax^2\cdot\frac
{1}{T}}, \]
where $\rho' = \sqrt[3]{T^{-\gamma} \log \frac{8Tk^3}{\delta}}$.

\paragraph{Step 1}
To prove \Cref{cor:aggregate_sample_complexity}, we first prove a variant of \Cref{lem:deviation-bound-of-regret-estimator}.

Let $\breve{x}^t(p)$ defined to be $\hat{x}^t(p)$ with $\pi^t$ is substituted with $\hat{\pi}^t$ estimated according \Cref{lem:aggregate_possible}.
Let
\[
    \breve{r}_{p,q}^t = \frac{1}{T}\left( \hat{\pi}^t(p)\left[(q-c)\breve{x}^t(q)-(p-c)\breve{x}^t(p)\right]\right), \widetilde{r}_{p,q}^t = \frac{1}{T}\left( \pi^t(p)\left[(q-c)\hat{x}^t(q)-(p-c)\hat{x}^t(p)\right]\right),
    \]
    and $\rho = \sqrt[3]{4 T^{-\gamma} \log \frac{2Tk}{\delta}}$ be the bound on the $\ell_{\infty}$ distance between the estimated $\hat{\pi}^t$ and $\pi^t$. 
    With probability $1-\delta$, we have
    \begin{align*}
        |\breve{r}_{p,q}^t - \widetilde{r}_{p,q}^t| &= \frac{1}{T}\left|(q-c)[\hat{\pi}^t(p)\breve{x}^t(q)-\pi^t(p)\hat{x}^t(q)]-(p-c)[\hat{\pi}^t(p)\breve{x}^t(p)-\pi^t(p)\hat{x}^t(p)]\right| \\
        &\leq \frac{1}{T}\left( \left|(q-c)[\hat{\pi}^t(p)\breve{x}^t(q)-\pi^t(p)\hat{x}^t(q)]\right| + \left|(p-c)[\hat{\pi}^t(p)\breve{x}^t(p)-\pi^t(p)\hat{x}^t(p)]\right|\right) \\
        & \leq \frac{1}{T}\pmax \left(\left|\hat{\pi}^t(p)\breve{x}^t(q)-\pi^t(p)\hat{x}^t(q)\right| + \left|\hat{\pi}^t(p)\breve{x}^t(p)-\pi^t(p)\hat{x}^t(p)\right| \right).
    \end{align*}

For the first term in the parenthesis we have
\begin{align*}
|\hat{\pi}^t(p)\breve{x}^t(q)-\pi^t(p)\hat{x}^t(q)| \leq& |\hat{\pi}^t(p)\breve{x}^t(q)-\hat{\pi}^t(p)\hat{x}^t(q)+\hat{\pi}^t(p)\hat{x}^t(q)-\pi^t(p)\hat{x}^t(q)| \\
\leq & |\hat{\pi}^t(p)\breve{x}^t(q)-\hat{\pi}^t(p)\hat{x}^t(q)| + |\hat{\pi}^t(p)\hat{x}^t(q)-\pi^t(p)\hat{x}^t(q)| \\
\leq & |\hat{\pi}^t(p)-\pi^t(p)|\hat{x}^t(q) + \hat{\pi}^t(p)|\breve{x}^t(q)-\hat{x}^t(q)|.
\end{align*}
Since we assume that the estimation in \Cref{lem:aggregate_possible} succeeded, we have $|\hat{\pi}^t(p)-\pi^t(p)|\leq \rho$. By assumption that $\pi^t(q)\geq \pimin$, we have $\hat{x}^t(q) \leq x^t(q)/\pi^t(q) \leq 1/\pimin$.  Therefore
\[
|\hat{\pi}^t(p)-\pi^t(p)|\hat{x}^t(q) \leq \frac{\rho}{\pimin}.
\]

When $T$ is large enough so that $\pimin > \rho$ holds, $\pi^t$ and $\hat{\pi}^t$ will have the same support with high probability. From now let us assume it is the case (and add this requirement on $T$ in the end as well). By definition of $\breve{x}^t$ and $\hat{x}^t$, for $q' = \max\{r\leq q : \hat{\pi}^t(r) > 0\},q'' = \max\{r\leq q:\pi^t(r) > 0\}$, we have 
\[
  \breve{x}^t(q) - \hat{x}^t(q) = \frac{x^t(q')}{\hat{\pi}^t(q')} - \frac{x^t(q'')}{\pi^t(q'')}.
\]
As $\hat{\pi}^t$ and $\pi^t$ have the same support, we know that
$q' = q''$. Therefore
\[
\breve{x}^t(q) - \hat{x}^t(q) \leq x^t(q') \left|\frac{\pi^t(q')-\hat{\pi}^t(q')}{\hat{\pi}^t(q')\pi^t(q')}\right| \leq 1 \cdot \frac{\rho}{(\pimin-\rho)\pimin }
\]
as $x^t(q')\leq 1$, $|\pi^t(q')-\hat{\pi}^t(q')|\leq \rho$, $\pi^t(q')\geq \pimin$, $\hat{\pi}^t(q') \geq \pi^t(q')-\rho$, and by the construction of $T$ that we have $\pimin > \rho$.

For the second term in the parenthesis, by definition of $\breve{x}^t,\hat{x}^t$ as well as the assumption that $\hat{\pi}^t$ and $\pi^t$ have the same support we have
\[
\left|\hat{\pi}^t(p)\breve{x}^t(p)-\pi^t(p)\hat{x}^t(p)\right| = 0. 
\]

Therefore, combining the two terms, with $1-\delta$ probability
\[
|\breve{r}_{p,q}^t - \widetilde{r}_{p,q}^t| \leq \frac{1}{T}\frac{\pmax \rho}{\pimin}\left(\frac
{1}{\pimin-\rho}+1\right).
\]

It follows that with probability at $1 - \delta$
\[
\AggEr(c,\allocseq{x}) - \EstEr(c,\allocseq{x}) \leq \frac{\pmax \rho}{\pimin}\left(\frac
{1}{\pimin-\rho}+1\right).
\]
Then by the triangle inequality and \Cref{lem:deviation-bound-of-regret-estimator}, we have, with probability at least $1-2 \delta$,

\[\AggEr(c,\allocseq{x}) - \WstEr(c,\allocseq{x}) \leq \frac{\pmax \rho}{\pimin}\left(\frac
{1}{\pimin-\rho}+1\right) + \sqrt{\log\frac{2}{\delta}\cdot 2 \left(\frac
{1}{\pimin}+1 \right)^2\pmax^2\cdot\frac
{1}{T}}.
\]
For the sake of brevity from now on we let
\[ f(T, \delta) = \frac{\pmax \rho}{\pimin}\left(\frac
{1}{\pimin-\rho}+1\right) + \sqrt{\log\frac{2}{\delta}\cdot 2 \left(\frac
{1}{\pimin}+1 \right)^2\pmax^2\cdot\frac
{1}{T}}. \]
Observe that $f(T, \delta)$ is weakly decreasing with $\delta$.

\paragraph{Step 2}
Next, we prove the first point of \Cref{cor:aggregate_sample_complexity}.

We claim that for any $c \in \costrange$
\[
\Pr \left[\AggEr(c,\allocseq{x})-\WstEr(c,\allocseq{x}) \geq kf(T,\delta) \right] \leq 2k^2\delta,
\]
which implies
\[
\Pr \left[\AggEr(c,\allocseq{x})-\WstEr(c,\allocseq{x}) \geq kf\left(T,\frac{\delta}{2k^2} \right) \right] \leq \delta.
\]
Monotonicity of $f(\delta, T)$ in $\delta$ further implies
\[
\Pr \left[\AggEr(c,\allocseq{x})-\WstEr(c,\allocseq{x}) \geq kf\left(T,\frac{\delta}{4k^2} \right) \right] \leq \delta.
\]

In fact,
\[
        \Pr\left[\AggEr(c, \allocseq{x}) - \WstEr(c, \allocseq{x}) \geq kf(T,\delta)\right] = \Pr\left[\sum_p \max_q \AggEr_{p, q}(c,\allocseq{x}) - \sum_p \max_q \WstEr_{p, q}(c,\allocseq{x})\geq kf(T,\delta)\right],
    \]
    and
     \begin{align*}
     &\Pr\left[\sum_p \max_q \AggEr_{p, q}(c,\allocseq{x}) - \sum_p \max_q \WstEr_{p, q}(c,\allocseq{x}) \geq kf(T,\delta)\right] 
     \\\leq &\Pr\left[\exists p \in \priceset, \exists p'\in \priceset,\AggEr_{p,p'}(c,\allocseq{x}) - \WstEr_{p,p'}(c,\allocseq{x}) \geq \frac{kf(T,\delta)}{|\priceset|} \right]
     \\\leq & 2k^2\delta &&      \mbox{(union bound)}.
    \end{align*}

Recall the choice of parameters: $\delta^T=kf\left(T, \frac{\delta}{4k^2} \right)$ and $T$ satisfies that $kf\left(T, \frac{\delta}{4k^2} \right) \leq \frac{r}{2}$. Suppose the seller satisfies $\TrueEr(c_0,\allocseq{x}) \leq r$, the probability that she fails the test
\begin{align*}
       &\Pr[\AggEr(\estplcost, \allocseq{x}) +\delta^T \geq 2r] \\ \leq 
       & \Pr[\AggEr(\estplcost, \allocseq{x}) +\delta^T \geq r + 2\delta^T] \\
       \leq &\Pr[\AggEr(\estplcost, \allocseq{x}) +\delta^T \geq \WstEr(c_0,\allocseq{x}) + 2\delta^T]
       \\
           \leq
           & \Pr[\AggEr(\estplcost, \allocseq{x}) - \WstEr(c_0,\allocseq{x}) \geq \delta^T] \\
           \leq
           & \Pr[\AggEr(c_0, \allocseq{x}) - \WstEr(c_0,\allocseq{x}) \geq \delta^T] && \text{(since $\AggEr(\estplcost, \allocseq{x}) \leq \EstEr(c_0, \allocseq{x})$)} \\
           \leq & \delta.
       \end{align*}
and the seller passes with probability at least $1-\delta$.

\paragraph{Step 3}
Finally, we prove the first point of \Cref{cor:aggregate_sample_complexity}.
We claim for $\estplcost = \argmin_{c \in \costrange}\AggEr(c, \allocseq{x})$ that
\[
\Pr \left[\AggEr(\plcost,\allocseq{x})-\WstEr(\plcost,\allocseq{x}) \leq -kf\left(T,\frac{\delta}{4k^2} \right) \right] \leq \delta.
\]

In fact, note that

\begin{align*}
       &\Pr\left[\AggEr(\estplcost, \allocseq{x}) - \WstEr(\estplcost,\allocseq{x}) \leq -kf(T,\delta)\right]
       \\\leq 
       &\Pr\left[\exists c,\AggEr(c, \allocseq{x}) - \WstEr(c,\allocseq{x}) \leq -kf(T,\delta)\right]
       \\=
       &\Pr\left[\exists c,\sum_p \max_q \AggEr_{p, q}(c,\allocseq{x}) - \sum_p \max_q \WstEr_{p, q}(c,\allocseq{x}) \leq -kf(T,\delta)\right]
       \\\leq
       &\Pr\left[\exists c,\exists p \in \priceset, \exists p'\in \priceset,\AggEr_{p,p'}(c,\allocseq{x}) - \WstEr_{p,p'}(c,\allocseq{x}) \leq -\frac{kf(T,\delta)}{|\priceset|} \right].
       \end{align*}
       Taking union bound over $p \in \priceset$ and $q' \in \priceset$, we have
        \begin{align*}   
        &\Pr\left[\exists c,\exists p \in \priceset, \exists q' \in \priceset, \AggEr_{p, p'}(c,\allocseq{x}) -  \WstEr_{p, p'}(c,\allocseq{x}) \leq -\frac{kf(T,\delta)}{|\priceset|}\right] 
        \\\leq
        & k^2\Pr\left[\exists c, \AggEr_{p, p'}(c,\allocseq{x}) -  \WstEr_{p, q'}(c,\allocseq{x}) \leq -\frac{kf(T,\delta)}{|\priceset|}\right].
       \end{align*}
       Observe that $\AggEr_{p, p'}(c,\allocseq{x}) -  \WstEr_{p, p'}(c,\allocseq{x})$ is linear in $c$ hence, when $c\in \costrange$,
       \begin{align*}
       &\Pr\left[\exists c, \AggEr_{p,p'}(c,\allocseq{x}) -  \WstEr_{p, p'}(c,\allocseq{x}) \leq -\frac{kf(T,\delta)}{|\priceset|}\right]\\\leq 
       &\Pr\left[\AggEr_{p, q'}(\ubar{c},\allocseq{x}) -  \WstEr_{p, q'}(\ubar{c},\allocseq{x}) \leq -\frac{kf(T,\delta)}{|\priceset|} \cup \AggEr_{p, q'}(\bar{c},\allocseq{x}) -  \WstEr_{p, q'}(\bar{c},\allocseq{x}) \leq -\frac{kf(T,\delta)}{|\priceset|}\right]
       \\\leq
       & \Pr\left[\AggEr_{p, q'}(\ubar{c},\allocseq{x}) -  \WstEr_{p, q'}(\ubar{c},\allocseq{x}) \leq - f(T,\delta)\right]+ \Pr\left[\AggEr_{p, q'}(\bar{c},\allocseq{x}) -  \WstEr_{p, q'}(\bar{c},\allocseq{x}) \leq -f(T,\delta)\right]\\
       \leq& 4\delta k^2.
       \end{align*}

       Recall the choice of parameters: $\delta^T=kf\left(T, \frac{\delta}{4k^2} \right)$ and $T$ satisfies that $kf\left(T, \frac{\delta}{4k^2} \right) \leq \frac{r}{2}$.
       Plugging in the bound of $T$ and $\delta^T$, we have 
       $\Pr\left[\AggEr(\estplcost, \allocseq{x}) - \WstEr(\estplcost,\allocseq{x}) \leq -\delta^T\right] \leq \delta$. 
       
       Therefore, when the seller's regret $\WstEr(\plcost,\allocseq{x}) \geq 2r$, by definition of $\plcost$ and $\estplcost$, we have
       \begin{align*}&\Pr\left[\AggEr(\estplcost, \allocseq{x}) +\delta^T \leq 2r\right] \\ \leq  &\Pr\left[\AggEr(\estplcost, \allocseq{x}) - 2r \leq -\delta^T\right] \\
       \leq
    &\Pr\left[\AggEr(\estplcost, \allocseq{x}) - \WstEr(\plcost,\allocseq{x}) \leq -\delta^T\right] \\ 
    \leq
    &\Pr\left[\AggEr(\estplcost, \allocseq{x}) - \WstEr(\estplcost,\allocseq{x}) \leq -\delta^T\right] && \text{(since $\WstEr(\plcost, \allocseq{x}) \leq \WstEr(\estplcost, \allocseq{x})$)} \\ 
       \leq &\delta.
    \end{align*}
       Hence the seller fails with probability at least $1-\delta$.

\paragraph{Step 4}
Finally, we do the algebra to verify our bound on $T$ indeed satisfies that $kf \left(T, \frac{\delta}{4k^2}\right) \leq \frac{r}{2}$ and $\pimin>\rho$. Recall that
\[ T \geq \max \left\{t_0, \left(\frac{4(8\pmax k +r\pimin)^3}{r^3\pimin^6}\right)^{2/\gamma}, \frac{16k^2}{r^2}\log \frac{8k^2}{\delta}\cdot\left(\frac{1}{\pimin}+1\right)^2\cdot \pmax^2, \left(\frac{4}{\pimin^3} \right)^{2/\gamma} \right\}+1, \]
where $t_0 = \sup \left\{t \in \mathbb R:t^\frac{\gamma}{2} \leq \log \frac{8tk^3}{\delta} \right\}$, and
\[ f \left(T, \frac{\delta}{4k^2} \right) = \frac{\pmax \rho}{\pimin}\left(\frac
{1}{\pimin-\rho}+1\right) + \sqrt{\log\frac{8k^2}{\delta}\cdot 2 \left(\frac
{1}{\pimin}+1 \right)^2\pmax^2\cdot\frac
{1}{T}}. \]

We show that in this case
\[ k\frac{\pmax \rho}{\pimin}\left(\frac
{1}{\pimin-\rho}+1\right) \leq \frac{r}{4}, \quad \text{and} \quad k\sqrt{\log\frac{8k^2}{\delta}\cdot 2 \left(\frac
{1}{\pimin}+1 \right)^2\pmax^2\cdot\frac
{1}{T}} \leq \frac{r}{4}. \]

In fact, since $T \geq t_0+1$, we have $T^\frac{\gamma}{2}>\log \frac{8Tk^3}{\delta}$, hence
\[ \rho \leq \sqrt[3]{\frac{4}{T^{\gamma/2}}} \leq \frac{r \pimin^2}{8 \pmax k+r \pimin}. \]
It follows that
\[ k\frac{\pmax \rho}{\pimin}\left(\frac
{1}{\pimin-\rho}+1\right) \leq k \frac{\pmax \rho}{\pimin} \cdot \frac{2}{\pimin-\rho} \leq k\frac{\pmax \frac{r \pimin^2}{8 \pmax k+r \pimin}}{\pimin} \cdot \frac{2}{\pimin-\frac{r \pimin^2}{8 \pmax k+r \pimin}} = k\frac{\pmax \frac{r \pimin^2}{8 \pmax k+r \pimin}}{\pimin} \cdot \frac{2}{\frac{8 \pmax k \pimin}{8 \pmax k +r \pimin}} = \frac{r}{4}. \]

The verification of the second $\leq \frac{r}{4}$ follows immediately using $T \geq \frac{16k^2}{r^2}\log \frac{8k^2}{\delta}\cdot\left(\frac{1}{\pimin}+1\right)^2\cdot \pmax^2$. The verification that $\pimin>\rho$ follows similarly from $T \geq \left(\frac{4}{\pimin^3} \right)^{2/\gamma}+1$ and $T \geq t_0+1$.

This concludes the proof.

\subsection{Proof of \texorpdfstring{\Cref{thm:manipulate}}{Theorem 5.2}}
\begin{proof}
Fix $\gamma=o(1)$. Let $\eps = \sqrt{\gamma}$. To prove the theorem we first provide the construction.
\begin{example}
    There are $(1+1.1)T$ rounds of interaction. The buyer's valuation $(v_1^t, v_2^t)$ is supported on $\{0, 1, 2, 3\} \times \{0, 1, 2, 3\}$ and the two sellers can post any price $p_i^t \in \mathcal P = \{0, 1, 2, 3\}$. Both sellers have cost $c_1=c_2=0$. The joint distribution of $(v_1^t, v_2^t)$ is shown in \Cref{table:no-external-regret-collusive-example} and
    \begin{table}[htbp]
        \renewcommand\arraystretch{1.25}
        \centering
        \begin{tabular}{|c|c|c|c|c|}
            \hline
            $v_1/v_2$ & 0 & 1 & 2 & 3 \\ \hline
            0 & 0 & 0 & 0 & $\frac{67}{600} + \frac{1}{3} \eps$ \\ \hline
            1 & 0 & 0 & 0 & $\frac{1}{30} - \frac{4}{3} \eps$ \\ \hline
            2 & 0 & 0 & 0 & $\frac{1}{100} + \eps$ \\ \hline
            3 & $\frac{1}{40}$ & $\frac{9}{25}$ & 0 & $\frac{23}{50}$ \\ \hline
        \end{tabular}
        \caption{No-best-in-hindsight-regret playing is not enough: value distribution}
        \label{table:no-external-regret-collusive-example}
    \end{table}
    is i.i.d. across rounds. We also assume that the buyer breaks ties randomly and she chooses to buy if buying gets utility 0.
\end{example}

We first note that with such a valuation, the buyer never chooses to buy nothing because either $v_1^t=3$ or $v_2^t=3$ with probability 1. It follows that
\begin{claim}
    Given prices $(p_1, p_2)$, the buyer buys good 1 if and only if $v_1-v_2>p_1-p_2$, buys good 2 if and only if $v_1-v_2<p_1-p_2$, and chooses randomly between seller 1 and 2 if $v_1-v_2=p_1-p_2$.
\end{claim}
\begin{proof}
    If the buyer buys good 1 then $v_1-v_2>p_1-p_2$ is necessary. If $v_1-v_2>p_1-p_2$ but she does not buy good 1, then $0>v_1-p_1>v_2-p_2$ so she buys nothing. But by construction, the buyer never does this.
\end{proof}
The claim enables us to write the demand function only using the distribution of $v_1-v_2$:
\begin{equation}\label{eqn:demand-manipulation}
    x_1(p_1, p_2) = \Pr[v_1-v_2>p_1-p_2] + \frac{1}{2} \Pr[v_1=v_2], \quad x_2(p_1, p_2)=1-x_1(p_1, p_2).
\end{equation}
Using \Cref{eqn:demand-manipulation} we construct the ex-ante payoff matrix in \Cref{table:no-external-regret-collusive-example-payoff} (note that playing price 0 is a dominated action so we omit it here).
\begin{table}[htbp]
    \renewcommand\arraystretch{1.25}
    \centering
    \begin{tabular}{|c|c|c|c|}
        \hline
        $p_1/p_2$ & 1 & 2 & 3 \\ \hline
        1 & $0.615, 0.385$ & $0.85 + \frac{\eps}{2}, 0.3 - \eps$ & $\frac{523}{600} + \frac{\eps}{3}, 0.385 - \eps$ \\ \hline
        2 & $0.77, 0.615$ & $1.23, 0.77$ & $1.7 + \eps, 0.45 - \frac{3 \eps}{2}$ \\ \hline
        3 & $0.615, 0.795$ & $1.155, 1.23$ & $1.845, 1.155$ \\ \hline
    \end{tabular}
    \caption{No-best-in-hindsight-regret playing is not enough: payoff matrix}
    \label{table:no-external-regret-collusive-example-payoff}
\end{table}
The highest payoff correlated equilibrium of this game is a pure NE where they play $(p_1, p_2)=(2, 2)$. The equilibrium payoff is $(1 + 1.1)T \cdot (1.23, 0.77)=(2.583T, 1.617T)$.

We claim there exists a manipulation such that the sellers play $(p_1, p_2)=(1, 1)$ in each round with high probability for $T-o(T)$ rounds, and then switch to collude by playing $(p_1, p_2)=(3, 3)$ in each round with high probability for $1.1T-o(T)$ rounds.

We first assume the claim is true. Then:
\begin{enumerate}
    \item Since $1.845>1.23, 1.155>0.77$, for $\Omega(T)$ rounds, both play $p_1>p_1^e, p_2>p_2^e$ in each round with high probability. This shows the first point of the theorem.
    \item By linearity of expectation, the total expected payoff is now $(T-o(T)) \cdot (0.615+1.1 \times 1.845, 0.385+1.1 \times 1.155)=(2.6445T-o(T), 1.6555T-o(T))>(2.583T, 1.617T)$, which is higher than the equilibrium payoff by $\Omega(T)$. This shows the second point of the theorem.
    \item Consider seller 1's best fixed action in hindsight:
    \begin{enumerate}
        \item If she plays price 1, the expected payoff is $(T-o(T)) \cdot 0.615+(1.1T-o(T)) \cdot (\frac{523}{600} + \frac{\eps}{3})<1.574T-o(T)$.
        \item If she plays price 2, the expected payoff is $(T-o(T)) \cdot 0.77+(1.1T-o(T)) \cdot (1.7 + \eps)=2.64T-o(T)$.
        \item If she plays price 3, the expected payoff is $(T-o(T)) \cdot 0.615+(1.1T-o(T)) \cdot 1.845=2.6445T-o(T)$.
    \end{enumerate}
    The best-in-hindsight price is 3 with an expected payoff of $2.6445T-o(T)$. But we just showed seller 1's total expected payoff in the manipulation is also $2.6445T-o(T)$, thus she has vanishing regret $o(T)$.
    
    Consider seller 2's best fixed action in hindsight:
    \begin{enumerate}
        \item If she plays price 1, the expected payoff is $(T-o(T)) \cdot 0.385+(1.1T-o(T)) \cdot 0.795<1.26T-o(T)$.
        \item If she plays price 2, the expected payoff is $(T-o(T)) \cdot (0.3 - \eps)+(1.1T-o(T)) \cdot 1.23=1.653T-o(T)$.
        \item If she plays price 3, the expected payoff is $(T-o(T)) \cdot (0.385 - \eps)+(1.1T-o(T)) \cdot 1.155=1.6555T-o(T)$.
    \end{enumerate}
    The best-in-hindsight price is 3 with an expected payoff of $1.6555T-o(T)$. But we just showed seller 2's total expected payoff in the manipulation is also $1.6555T-o(T)$, thus she has vanishing regret $o(T)$.
\end{enumerate}
Note that seller 1 has non-vanishing calibrated regret because the best response to price 1 is price 2, and seller 2 has non-vanishing calibrated regret because the best response to price 3 is price 2.

Next, we show the claim: How to manipulate a $\gamma$-mean-based seller 2 into a collusion under this setting.
    
Seller 1 manipulates as follows: She first plays $p_1=1$ for $T$ rounds, and then switches to playing $p_1=3$ for the remaining $1.1T$ rounds.

The following claim follows the definition of $\gamma$-mean-based strategy.
\begin{claim}
    With a $\gamma$-mean-based algorithm:
    \begin{enumerate}
        \item For each $T \geq t \geq O(\eps T)$, seller 2 posts $p_2=1$ in round $t$ w.p. at least $1 - \gamma$.
        \item For each $T+O(\eps T) \leq t \leq T+1.1T$, seller 2 posts $p_2=3$ in round $t$ w.p. at least $1 - \gamma$.
    \end{enumerate}
\end{claim}
\begin{proof}
    We write the cumulative reward of playing prices 1, 2, and 3 as follows:
    \begin{align*}
        r_1(t) &= \begin{cases} 0.385t & t \leq T \\ 0.385T+0.795(t-T) & t \geq T \end{cases}, \\
        r_2(t) &= \begin{cases} (0.3 - \eps)t & t \leq T \\ (0.3 - \eps)T+1.23(t-T) & t \geq T \end{cases}, \\
        r_3(t) &= \begin{cases} (0.385 - \eps)t & t \leq T \\ (0.385 - \eps)T+1.155(t-T) & t \geq T \end{cases}.
    \end{align*}
    It follows that between $t = \eps(1+1.1)T$ and $t=T$ we have $r_1(t) \geq \max(r_2(t), r_3(t)) + \gamma T$. Just after $t=T$ price 1 has an advantage of $\eps T$ but price 3 quickly comes as the best choice after $T + \eps T/(1.155-0.795) \leq T+3 \eps T$. Price 3 still dominates until $t=0.085T/(1.23-1.155)<1.1T$ (this is when price 2 becomes the best). Therefore, for each $T \geq t \geq O(\sqrt \gamma T)$ seller 2 posts $p_2=1$ w.p. at least $1 - \gamma$, and for each $T+O(\eps T) \leq t \leq 1.1T+T$ seller 2 posts $p_2=3$ w.p. at least $1 - \gamma$.
\end{proof}
\begin{remark}
    Although the above argument suffices for our purposes, we remark that the claim still holds even if the sellers only get a realization of the buyer's decision each round (instead of getting the expected reward of her strategy) by using a concentration argument and the length of the time window where the sellers are supra-competitive will only suffer an $o(T)$ loss.
\end{remark}
This ends the proof.
\end{proof}

\begin{algorithm}[htb]
    \KwData{$\{p^t\}_{t=1}^T, \{x^t(p^t)\}_{t=1}^T, \{(\pi^t)\}_{t=1}^T$}
    \For{$t \in [T]$}{
        $C^t \ASGN \{p \in \priceset:\pi^t(p_j)>0\}$;
        
        \For{$p \in C^t$}{
            $\hat x^t(p) \ASGN
            \begin{cases}
                x^t(p)/\pi^t(p), & p=p^t, \\
                0, & \text{otherwise}
            \end{cases}
            $;
        }
        \For{$p \in \priceset$}{
            $\hat{h}^t(p) \ASGN \hat x^t(p')$ where $p' = \max\{r \leq p:r \in C^t\}$;
        }
    }
    Solve the following programming and let the solution be $\estplcost$, defined as \emph{estimated plausible cost}:
    \begin{equation*}
        \min_{c \in \costrange} \quad \EstEr(c)
    \end{equation*}
    where
    \begin{gather*}
        \EstEr(c) = \sum_p \max_q \EstEr_{p, q}(c), \\
        \EstEr_{p, q}(c) = \frac 1T \sum_{t=1}^T \pi^t(p)\left[(q-c)\hat h^t(q)-(p-c)\hat h^t(p)\right].
    \end{gather*}
    
    Let 
    \begin{equation*}
        \delta^T = \frac{k\pmax}{T}\sqrt{2\log\left(\frac{2k^2}{\alpha}\right)\cdot \sum_{s=1}^T\left(\frac{1}{\min_{p \in C^s}\pi^s(p)}+1\right)^2}.
    \end{equation*}
 
 \eIf{$\EstEr(\estplcost) + \delta^T \leq 2r$}
    {
       Output PASS;
    }
    {
       Output FAIL;
    }
    
    \caption{Auditing (non)-collusion via testing the pessimistic regret}
    \label{algo:detection}
\end{algorithm}

\end{document}